\documentclass[11pt]{article}
\usepackage{mydef2col}
\usepackage{array}
\usepackage{bm}

\makeatletter
\def\blfootnote{\xdef\@thefnmark{}\@footnotetext}
\makeatother
\providecommand{\keywords}[1]{\blfootnote{\small \textbf{Keywords:} #1}}
\providecommand{\codes}[1]{\blfootnote{\small \textbf{2010 MSC: } #1}}

\allowdisplaybreaks[4]
\graphicspath{{./Graphs/}}
\usepackage{tikz}
\usetikzlibrary{calc,decorations.markings}
\usetikzlibrary{arrows,patterns,positioning}

\usepackage[round]{natbib}

\def\calJ{{\cal J}}
\def\calB{{\cal B}}

\def\baru{{\bar{u}}}
\def\barU{{\bar{U}}}

\def\sp{\sqrt{p} }

\newcommand{\iu}{\mathrm{i}\mkern1mu}

\newcommand{\Ree}{\operatorname{Re}}

\newcommand{\Imm}{\operatorname{Im}}

\title{Semi-analytical pricing of barrier options in the time-dependent Heston model}

\author{
\authorstyle{
Peter Carr{}
\textsuperscript{1}
, Andrey Itkin{}
\textsuperscript{1}
and Dmitry Muravey
\textsuperscript{2}
}
\newline\newline
\textsuperscript{1}
\institution{Tandon School of Engineering, New York University, 1 Metro Tech Center, 10th floor, Brooklyn NY 11201, USA} \\
\textsuperscript{2}
\institution{Moscow State University, Moscow, Russia}
}

\vspace{-0.2in}
\date{\today}
\begin{document}

\maketitle
\vspace{-0.2in}

\lettrineabstract{We develop the general integral transforms (GIT) method for pricing barrier options in the time-dependent Heston model (also with a time-dependent barrier) where the option price is represented in a semi-analytical form as a two-dimensional integral. This integral depends on yet unknown function $\Phi(t,v)$ which is the gradient of the solution at the moving boundary $S = L(t)$ and solves a linear mixed Volterra-Fredholm equation of the second kind also derived in the paper. Thus, we generalize the one-dimensional GIT method, developed in (Itkin, Lipton, Muravey, Generalized integral transforms in mathematical finance, WS, 2021) and the corresponding papers, to the two-dimensional case. In other words, we show that the GIT method can be extended to stochastic volatility models (two drivers with inhomogeneous correlation). As such, this 2D approach naturally inherits all advantages of the corresponding 1D methods, in particular, their speed and accuracy. This result is new and has various applications not just in finance but also in physics. Numerical examples illustrate high speed and accuracy of the method as compared with the finite-difference approach.} \hspace{10pt}


\codes{91G40; 91G60; 91G80; 47G30; 35R09; 65L12;}


\section*{Introduction} \label{layer}

The classical Heston model was introduced in \citep{Heston:1993a}. It immediately drew a lot of attention since the characteristic function (CF) of the log-spot price $x_t$ for this model can be found in closed form. Thus, the pricing of European options in this model becomes almost straightforward by using the well-known FFT methods, see, e.g., survey in \citep{Schmelzle2010}. The Heston model belongs to the class of stochastic volatility (SV) models and introduces an instantaneous variance $v_t$ as a mean-reverting square-root process correlated to the underlying stock price process $S_t$. The model is defined by the following stochastic differential equations (SDEs):
\begin{align} \label{heston}
dS_t &= S_t(r-q) dt + S_t\sqrt{v}dW^{(1)}_t, \\
d v_t &= \kappa (\theta - v_t)dt + \sigma\sqrt{v_t} dW^{(2)}_t, \nonumber \\
d \langle W^{(1)}_t, W^{(2)}_t \rangle &= \rho dt, \quad [S,t] \in [0,\infty) \times [0,\infty),
\quad S_0 = S, \quad v_0 = v, \nonumber
\end{align}
\noindent where $W^{(1)}$ and $W^{(2)}$ are two standard correlated Brownian motions with the constant correlation coefficient $\rho \in [-1,1]$, $\kappa$ is the rate of mean-reversion, $\sigma$ is the volatility of variance $v_t$ (vol-of-vol), $\theta$ is the mean-reversion level (the long-term run), $r$ is the interest rate and $q$ is the continuous dividend. All parameters in the original Heston model are assumed to be time independent. If the so-called Feller condition $2 \kappa \theta > \sigma^2$ is satisfied, the process $v_t$ is strictly positive, $v_t \in [0,\infty)$; otherwise its behavior at the origin should be additionally identified, see e.g., \citep{f54,CarrLinetsky2006,Lucic2008}.

Despite high popularity of the Heston model among both practitioners and researchers, later it was observed, \citep{Benhamou2010}, that using this model is still a challenge because the CF is known in closed form only when the parameters are constant or piecewise constant, \citep{Mikhailov2003,Guterding2018}. And the time dependence of the parameters is necessary to be able to calibrate the model to the term-structure of market data. Thus, for the time dependent parameters there is no an any analytical formula for the European option price, and one usually has to perform either a Monte Carlo simulation, \citep{AndersenHeston} and references therein, or use a finite-difference (FD) approach, \citep{Kluge2002,ItkinCarrBarrierR3} and references therein. To improve this, in \citep{Benhamou2010}
a small volatility of volatility expansion and Malliavin calculus techniques are used to derive an analytic approximation
for the price of vanilla options for any time dependent Heston model. A survey of various approaches to pricing options under the time-dependent Heston model can be found in \citep{Rouah2015}. The time-dependent correlation function was also considered in \citep{Teng2021}.

However, for exotic options, such as e.g., barrier options, not so many analytical results have been obtained even for the case of constant coefficients, where they are available in two basic cases. The first one is a zero drift and zero correlation case where the spot and instantaneous variance processes are uncorrelated. Therefore, conditional on the integrated variance the option price is given by a 1D formula for the corresponding barrier option which should be further integrated with the density of the integrated variance, see \citep{LiptonMcGhee2002, Lipton2001}.  Often, this density is known in closed form, but for some models only the Laplace transform of the density is known that brings additional complexity (as applied to the CIR variance process, see \citep{ContTankov, Belomestny2016}). Aside of technical problems, this approach potentially can be further applied to the time dependent SV model with no correlation and drift.

The other tractable case is when the model has a small parameter $\epsilon$, so the solution can be constructed asymptotically by using a series expansion in $\epsilon$, see \citep{LiptonMcGhee2002, Sircar2004, Kato2013, LiptonGal2014, Lorig2017} among others. For instance, in \citep{LiptonMcGhee2002} this is done by assuming that $v T \ll 1$. However, for the time-dependent model  construction of such semi-analytical solution could become problematic.

An attempt to make the next step has been done in a recent paper \citep{Aquino2019} where the barrier option price in the Heston model with constant coefficients has been presented in a semi-analytical form. The authors tried to  extend the approach of \citep{Griebsch2013} by using conditioning on the variance path and then employing the reflection properties of the Brownian motion. They claim that the stock price at maturity $S_T$ conditional on $v_T$ and the integrated variance $\bar{v} = \int_0^T v_s d s$ has a lognormal distribution, and then construct  the joint probability distribution (pdf) of the logarithmic spot price and its maximum/minimum by using the reflection principle. Finally, they derive a joint pdf for $v_T$ and $\bar{v}$ via a double inverse Fourier transform. Unfortunately, as was figured out by Prof. A. Lipton during our joint discussions, and later confirmed by the authors of  \citep{Aquino2019}, their derivation contains an error, and so their final result should be discarded.

Even if it had been correct,  the reflection principle would not have been valid for time-dependent barriers. Also, time-dependent coefficients of the model make it hard to derive both the joint pdf mentioned in above (for the joint pdf of $v_T$ and $\bar{v}$ perhaps, this is possible if only some coefficients are functions of time while the other are constant, e.g. if $\kappa, \rho$ are constants as in \citep{CarrSun}). Finally, as this will be seen below, numerical complexity of this approach is close to that proposed in this paper (computation of two-dimensional complex non-singular integrals).

Therefore, practitioners who need to price barrier options using the whole time-dependent Heston model with no simplifications yet have to use numerical methods. In this paper we develop an alternative approach to this problem by using the generalized integral transform method (GIT) originally developed in physics and then introduced into mathematical finance by the authors in \citep{CarrItkin2020jd, ItkinMuravey2020r, CarrItkinMuravey2020, ItkinMuraveyDBFMF} and also in cooperation with Alex Lipton in \citep{ItkinLiptonMuravey2020,ItkinLiptonMuraveyMulti,ItkinLiptonMuraveyBook}. To shorten the references, in what follows we cite just a recent book, \citep{ItkinLiptonMuraveyBook}, having in mind that the corresponding materials could also be found in the above referenced papers.

Despite our methods can be applied to any sort of barrier options, here, as an example, we consider only a Down-and-Out barrier Put option written on the underlying process $S_t \in [L(t), \infty]$, which follows the dynamics in \eqref{heston} with all the model coefficients $\kappa, \theta, \sigma, \rho$ being functions of the time $t$, and where $L(t) > 0$ is the lower barrier. We also discuss other types of the barrier options in Section~\ref{Discus}.

We assume that once $S_t$ hits the barrier, the contract is terminated and the option expires worthless, i.e.
\begin{equation} \label{bcH}
P(t,L(t),v) = 0,
\end{equation}
\noindent where $P(t,S,v)$ is the option price. In other words, in this case we assume no rebate is paid either at the option maturity $T$, or at hit. This assumption can be easily relaxed, see \citep{ItkinMuraveyDBFMF}. At the other boundary we assume the standard condition
\begin{equation} \label{bc0}
P(t,S,v)\Big|_{S \uparrow \infty} = 0.
\end{equation}

If the process $S_t$ survives till $t=T$, the Put option provides its holder with the payoff
\begin{equation} \label{tc}
P(T,S,v) = (K - S)^+,
\end{equation}
\noindent where $K > 0$ is the strike. The \eqref{tc} is the terminal condition for our problem. We also assume that $L(T) <  K$.

Our main result obtained in this paper is as follows. We develop the GIT method for pricing barrier options in the time-dependent Heston model (also with the time-dependent barrier) and derive a semi-analytical solution of this problem which is expressed via a two-dimensional integral. This integral depends on yet unknown function $\bar{\Phi}(t,v)$ which is the gradient of the solution at the moving boundary $S = L(t)$ and solves a linear mixed Volterra-Fredholm (LMVF) equation of the second kind also derived in the paper. Briefly speaking, we generalize a one-dimensional GIT method developed in \citep{ItkinMuravey2020r} to the two-dimensional case. Or, to say it differently, we show that the GIT method can be developed not only for one-factor models, but for the SV models (two drivers with inhomogeneous correlation) as well. As such, this 2D method naturally inherits all advantages of the corresponding 1D methods, in particular, their speed and accuracy. This result is new and has various applications not just in finance but in physics as well.

The rest of the paper is organized as follows. In Section~\ref{pricePDE} we consider a partial differential equation (PDE) for the price of Down-and-Out Put option and solve it by using generalization of our GIT method. In Section~\ref{solVolt} we discuss how the LMVF equation derived in Section~\ref{pricePDE} can be solved numerically. In doing so we use the method of Radial Basis Functions (RBF) but replace a Gaussian RBF with another one. This new basis function (BF) is actually not an RBF but mimics the Gaussian RBF and is positive-definite, hence can be used as an interpolation kernel. We prove all these properties of the new BF in Appendix~\ref{appPD}. The main idea of the new BF is that it makes the problem tractable by reducing the 3D integral in the LMVF equation to the 2D one. Section~\ref{numExp} describes results of our numerical experiments where the prices of barrier options are obtained by using the GIT method and then compared with those computed by using a finite difference (FD) approach. We show that our method outperforms the FD one in both accuracy and speed. Section~\ref{Discus} concludes and provides some additional comments about capability of the developed approach.

\section{The pricing PDE and its solution} \label{pricePDE}

Let us introduce a new variable $x = \log(S/K)$. By the standard argument, \citep{ContVolchkova2005}, under the risk neutral measure the Put option price $P(t,x,v)$ with $x,v$ being the initial values of processes $x_t, v_t$ at the time $t=0$  solves the partial differential equation (PDE)
\begin{align} \label{PDE1}
\fp{P}{t} &+ \dfrac{1}{2} v \sop{P}{x}  +  \left[r(t) - q(t) - \frac{1}{2} v\right] \fp{P}{x}
+ \dfrac{1}{2} \sigma^2(t) v \sop{P}{v} + \kappa(t) (\theta(t) - v)\fp{P}{v}  + \rho(t) \sigma(t) v \cp{P}{x}{v} = r(t) P,
\end{align}
\noindent subject to the terminal condition
\begin{equation} \label{tc1}
P(T,x,v) = K(1 - e^x)^+,
\end{equation}
\noindent and the boundary conditions
\begin{align} \label{bcH1}
P(t,y(t), v) &= 0, \quad y(t) = \log(L(t)/K) < 0, \\
P(t,x, v)\Big|_{x \uparrow \infty} &= 0. \nonumber
\end{align}

Following the idea of the method of generalized integral transforms (GIT) for $S \in [L(t), \infty)$, \citep{ItkinMuravey2020r}, we introduce the following integral transform
\begin{equation} \label{GITdef}
\baru(\tau,p,v) = \int_{y(t)}^\infty P(t,x,v) e^{-\sqrt{p}x} dx,
\end{equation}
\noindent where $p = a + \iu\omega$ is a complex number. It might look that we also need to request $\Ree(p) = \beta > 0$ for the transform to exist. However, usually the solution $u(t,x,v)$ converges to zero as $u(t,x,v) \propto e^{ - a x^2}, \ a > 0$, see e.g., \citep{ItkinMuravey2020r}, hence the integral in the RHS of \eqref{GITdef} is well-behaved.

Then, multiplying both parts of \eqref{PDE1} by $e^{-x\sqrt{p}}$ and integrating on $x$ from $y(t)$ to infinity, we obtain
\begin{align} \label{tr2}
0 &= \frac{\partial}{\partial t} \baru + P(t,y(t),v) e^{-\sqrt{p}y(t)} y'(t) + \left[r(t) - q(t) - \frac{1}{2}v\right] \left[e^{-x \sqrt{p}} P(t,x,v) \Big|_{x = y(t)}^{\infty} + \sqrt{p} \baru \right] \\
&+ \dfrac{1}{2} v \left( e^{-x \sqrt{p}} P_x(t,x,v) \Big|_{x = y(t)}^{\infty}
+ \sqrt{p} e^{-x \sqrt{p}} P(t,x,v) \Big|_{x = y(t)}^{\infty} + p \baru \right)  \nonumber \\
&+ \dfrac{1}{2} \sigma^2(t) v \sop{\baru}{v} + \kappa(t) (\theta(t) - v)\fp{\baru}{v} +  \rho(t) \sigma(t) v \sqrt{p} \fp{\baru}{v} + \rho(t) \sigma(t) v e^{-x \sqrt{p}} P_v(t,x,v) \Big|_{x = y(t)}^{\infty} - r(t) \baru. \nonumber
\end{align}
With allowance for the boundary conditions, \eqref{tr2} reduces to
\begin{align} \label{tr3}
0 &= \frac{\partial}{\partial t} \baru  + \left[r(t) - q(t) - \frac{1}{2}v\right] \sqrt{p} \baru + \dfrac{1}{2} v \left( -e^{-y(t) \sqrt{p}} P_x(t,y(t),v) + p \baru \right)  \\
&+ \dfrac{1}{2} \sigma^2(t) v \sop{\baru}{v} + \kappa(t) (\theta(t) - v)\fp{\baru}{v} +  \rho(t) \sigma(t) v \sqrt{p} \fp{\baru}{v} - \rho(t) \sigma(t) v e^{- y(t) \sqrt{p}} P_v(t, y(t),v) - r(t) \baru, \nonumber
\end{align}
\noindent or, after some algebra,
\begin{align} \label{PDEv}
0 = \frac{\partial}{\partial t} \baru  &+ [a(t,p) + c(p) v ]\baru(t,v)  + \dfrac{1}{2} \sigma^2(t) v \sop{\baru}{v} + \bar{\kappa}(t,p)(\bar{\theta}(t,p) - v) \fp{\baru}{v} - v e^{-y(t) \sqrt{p}} \Phi(t,v), \\
a(t,p) &= r(t)(\sqrt{p} -1)  - q(t)\sqrt{p}, \quad c(p) = \frac{1}{2}(p - \sqrt{p}), \quad
\bar{\kappa}(t,p) = \kappa(t) - \rho(t) \sigma(t) \sqrt{p}, \nonumber \\
\bar{\theta}(t,p) &= \theta(t) \frac{\kappa(t)}{\bar{\kappa}(t,p)}, \quad
\Phi(t,v) =  \dfrac{1}{2} P_x(t,y(t),v) + \rho(t) \sigma(t) P_v(t, y(t),v), \nonumber \\
\baru(T, p) &= K \left[ \frac{ e^{-y(T) \sqrt{p}} - 1}{\sqrt{p}} - \frac{e^{-y(T) \left(\sqrt{p}-1\right)} - 1}{\sqrt{p}-1} \right].  \nonumber
\end{align}

Assuming that the function $u(t,x,v)$ is smooth enough at the boundary $x = y(t)$, it follows that
\begin{equation}
\lim_{x \to y(t)} P_v(t,x,v) = \partial_v \lim_{x \to y(t)} P(t,x,v) = \partial_v P(t,y(t),v) = 0,
\end{equation}
\noindent and, hence, the second term in the definition of $\Phi(t,v)$ vanishes.

\subsection{Solution of \eqref{PDEv}}

The \eqref{PDEv} is an inhomogeneous PDE and also exponentially affine in the variable $v$. Its solution can be constructed if the Green's function of the homogeneous PDE is known. It can be observed that a similar homogeneous PDE is considered in \citep{CarrItkinMuravey2020} with respect to pricing barrier options under the CIR model. Therefore, we can proceed in the same way.
\begin{proposition} \label{prop2}
The \eqref{PDEv} can be transformed to the form
\begin{equation} \label{Bess1}
\fp{\barU}{\tau} = \frac{1}{2} \sop{\barU}{z} + \frac{b}{z} \fp{\barU}{z} + \bar{\Psi}(\tau, z),
\end{equation}
\noindent where $b$ is some constant, $\barU = \barU(\tau, z)$ is the new dependent variable, and $(\tau, z)
\in [0,\infty) \times [0,\infty)$ are the new independent variables, if
\begin{equation} \label{cond1}
\frac{\kappa(t) \theta(t)}{\sigma^2(t)} = \frac{m}{2},
\end{equation}
\noindent where $m \in [0,\infty)$ is some constant. The homogeneous version of \eqref{Bess1} is the PDE associated with the one-dimensional Bessel process, \citep{RevuzYor1999}
\begin{equation} \label{BesProc}
d X_t = d W_t  + \frac{b}{X_t} dt.
\end{equation}

\begin{proof}[{\bf Proof}] First make a change of variables
\begin{alignat}{2} \label{trCIR1}
\baru(t,v) &= \barU(t,z) e^{\alpha(t,p) v + \beta(t,p)}, \qquad & z &= g(t,p) \sqrt{v}, \\
g(t,p) &=  \exp\left[ \frac{1}{2} \int_0^t \left( \bar{\kappa}(s,p) - \alpha(s,p) \sigma^2(s) \right) \, ds \right], \qquad & \beta(t,p) &= \int _t^T [a(s,p) + \kappa(s) \theta(s) \alpha (s,p)] ds, \nonumber
\end{alignat}
\noindent where $\alpha(t,p)$ solves the Riccati equation
\begin{equation} \label{ric2}
\alpha'(t,p) = - c(p) + \bar{\kappa}(t,p) \alpha (t,p)  - \frac{1}{2}\alpha(t,p)^2 \sigma (t)^2.
\end{equation}

In new variables the PDE \eqref{PDEv} reads
\begin{align} \label{PDE12}
\frac{4 \bar{\kappa}(t,p) \bar{\theta}(t,p) - \sigma^2(t)}{2 z} \fp{\barU}{z} &+ \frac{1}{2} \sigma (t)^2 \sop{\barU}{z} + \frac{4}{g^2(t,p)}\fp{\barU}{t} + \Phi_1(t,z,p) = 0, \\
\Phi_1(t,z,p) &= - \left[\frac{4 v}{g^2(t,p)} e^{-y(t) \sqrt{p}}  e^{-[\beta (t,p) + v \alpha (t,p)]} \Phi(t,v)\right]_{v \to z^2/g(t,p)^2}. \nonumber
\end{align}
Next, by introducing the backward time $\tau$
\begin{equation} \label{tauTr}
\tau(t,p) = \frac{1}{4} \int_t^T  g^2(s,p) \sigma^2(s) \, ds,
\end{equation}
\noindent we convert \eqref{PDEv} to the form
\begin{align} \label{PDE3}
\fp{\barU}{\tau} &= \left[2 \frac{\bar{\kappa}(t) \bar{\theta}(t)}{\sigma^2(t)} - \frac{1}{2}\right] \frac{1}{z} \fp{\barU}{z} + \frac{1}{2}\sop{\barU}{z} + \Psi(t,z,p), \quad t = t(\tau), \quad \Psi(t,z,p) = \frac{\Phi_1(t,z,p)}{\sigma^2(t)}.
\end{align}

The function $t(\tau,p)$ is the inverse map of \eqref{tauTr}. Given the value of $p$, this map can be computed for any $\tau \ge 0$ by using the definition in \eqref{tauTr} and then inverting.

Finally, using the assumption in \eqref{cond1} we get
\begin{equation*}
\frac{\bar{\kappa}(t) \bar{\theta}(t)}{\sigma^2(t)} = \frac{\kappa(t) \theta(t)}{\sigma^2(t)} = \frac{m}{2},
\end{equation*}
\noindent and set the constant $b = m - 1/2$. Hence, the proposition is proved.
\end{proof}
\end{proposition}

As this is mentioned in \citep{CarrItkinMuravey2020} and follows from Proposition~\ref{prop2}, for the Heston model the transformation from \eqref{PDEv} to \eqref{PDE3} cannot be done unconditionally. However, even with the restriction in \eqref{cond1} the model still makes sense. Indeed, the model parameters already contain the independent mean-reversion rate $\kappa(t)$ and vol-of-vol $\sigma(t)$.  Since $m$ is an arbitrary constant, it could be calibrated to the market data together with $\kappa(t)$ and $\sigma(t)$. Therefore, even in this form the Heston model should be capable to be calibrated to the market option prices.

The terminal condition in \eqref{PDEv} doesn't depend on $v$ which means
\begin{equation}
\fp{}{v}\left[\barU(0,z) e^{\alpha(T,p) v + \beta(T,p)}\right] = e^{\alpha(T,p) v+ \beta(T,p)} \left[\fp{\barU(0,z)}{z} \frac{g^2(T,p)}{2 z}  + \barU(0,z) \alpha(T,p) \right] = 0,
\end{equation}
\noindent or
\begin{equation} \label{tcU}
\barU(0,z) = \baru(T,p) e^{- B z^2}, \qquad B  = \alpha(T,p)/g^2(T,p).
\end{equation}
\noindent where $\baru(T,p)$ has been determined in \eqref{PDEv}.

Since \eqref{Bess1} is an inhomogeneous PDE, it can be solved if the Green's function of the corresponding homogeneous PDE is known. Since this homogeneous counterpart with $\Psi(\tau, z,p) = 0$ is the Bessel equation defined at the semi-infinite domain $z \in [0,\infty)$, this Green's function is known in closed form assumed that the Bessel process stops when it reaches the origin. In more detail, it is relatively easy to show that the boundary $v_t = 0$ is an attainable regular boundary by Feller's classification, \citep{Lipton2001}. Therefore, similar to \citep{GorovoiLinetsky} we always make regular boundaries instantaneously reflecting, and include  regular reflecting boundaries into the state space.  We also assume that infinite boundaries are unattainable..

Since by definition $ m \in [0,\infty)$, one has to consider two cases determined by the famous Feller's condition. If the Feller condition is satisfied and the process never hits the origin (which means that $m \ge 1$), by the definition of $b$ in Proposition~\ref{prop2} this implies $b \ge 1/2$. It is known, \citep{Lawler2018NotesOT,LinetskyMendozza2010},  that in case $b \ge 1/2$ the density $G(\tau, z,\zeta)$ is a good density with no defect of mass, i.e., it integrates into 1. The explicit representation reads, \citep{c75, Emanuel:1982}
\begin{equation} \label{Green}
G(\tau, z,\zeta) = \frac{\sqrt{z \zeta}}{\tau} \left(\frac{\zeta}{z}\right)^b  e^{- \frac{z^2 + \zeta^2}{2 \tau}} I_{b-1/2} \left( \frac{z \zeta}{\tau}\right).
\end{equation}
Here $I_\nu(x)$ is the modified Bessel function of the first kind, \citep{as64}.

Otherwise, if $0 < m < 1$ this implies $-1/2 < b < 1/2$. Then by another change of variables, \citep{Polyanin2002}
\begin{equation*}
\barU(\tau,z) = z^{2(1-m)} \barU_1(\tau,z),
\end{equation*}
\noindent the \eqref{PDE3} transforms to the same equation with respect to $\barU_1(\tau,z)$ but now with $b = (3-2m)/2$. Accordingly, since $0 < m < 1$ we have $b > 1/2$, Therefore, again the Green's function is represented by \eqref{Green}.

Since the Green's function of the homogeneous form of \eqref{Bess1} is known, the solution of \eqref{PDE3} can be represented as, \citep{Polyanin2002}
\begin{align} \label{sol}
\barU(\tau,z) &= \int_0^\infty \barU(0,\zeta) G(\tau, z, \zeta) d\zeta + \int_0^\tau \int_0^\infty G(\tau-s, z, \zeta) \Psi(s,\zeta) ds \, d\zeta.
\end{align}
Using the definition of $\barU(0,z)$ in \eqref{tcU}, the first integral can be computed in closed form
\begin{align} \label{int1}
I_1(\tau,z,p) \equiv \int_0^\infty \barU(0,\zeta) G(\tau, z, \zeta) d\zeta &= K \left[ \frac{ e^{-y(T) \sqrt{p}} - 1}{\sqrt{p}} - \frac{e^{-y(T) \left(\sqrt{p}-1\right)} - 1}{\sqrt{p}-1} \right] \frac{e^{-\frac{B z^2}{2 B \tau +1}}}{(2 B \tau +1)^{b+\frac{1}{2}}},
\end{align}

Returning back to the original variables $(t,v)$ (in more detail see Appendix ~\ref{app2}) we obtain from \eqref{sol}
\begin{align} \label{sol1}
\baru(t,v,p) &= e^{\alpha(t,p) v + \beta(t,p)} I_1(t,v,p) \\
&+ \frac{1}{2} e^{\alpha(t,p) v + \beta(t,p)} \int_t^T \int_0^\infty \sqrt{v'} g(s,p) \Bigg\{ G\left(\int_t^s \frac{1}{4} g^2(\gamma, p) \sigma^2(\gamma) d\gamma, g(t,p) \sqrt{v}, g(s, p) \sqrt{v'} \right) \nonumber \\
&\times e^{-y(s)\sqrt{p} - \left[\beta(s, p) + v'\alpha(s,p)\right]} 	\Phi\left(s, v' \right) \Bigg\}ds \, dv'. \nonumber
\end{align}

\subsection{Solution of the Riccati equation \eqref{ric2}} \label{secLMVF}

In case the model coefficients are time-homogeneous, i.e. $\bar{\kappa}(t,\sp) = \bar{\kappa}(\sp), \ \sigma(t) = \sigma$, \eqref{ric2} subject to the terminal condition $\alpha(T,p) = \alpha_T$ can be solved analytically. The solution reads
\begin{align} \label{solRicConst}
\alpha(t,\xi) &= \frac{1}{\sigma^2} \Bigg\{\bar{\kappa}(\sp) + \sqrt{2 c((\sp)) \sigma^2 - \bar{\kappa}^2(\sp)} \\ &\times \tan \left[\tan^{-1}\left(\frac{\alpha_T  \sigma^2 - \bar{\kappa}(\sp)} {\sqrt{2 c(\sp) \sigma^2 - \bar{\kappa}^2(\sp)}}\right) + \frac{1}{2} (T-t) \sqrt{2 c(\sp) \sigma^2 - \bar{\kappa}^2(\sp)}\right] \Bigg\}  \nonumber \\
&= \frac{ C(t,p) \alpha(T,\xi)  + \left[2 c(\sp)  - \bar{\kappa}(\sp) \alpha(T,\xi) \right] \tan \left(C(t,p) (T-t)/2\right) }{C(t,p)  + \left[\bar{\kappa}(\sp) - \alpha(T,\xi) \sigma^2 \right] \tan \left(C(t,p) (T-t)/2\right)},
\quad C(t,p) = \sqrt{2 c(\sp) \sigma^2 - \bar{\kappa}^2(\sp)}, \nonumber
\end{align}
\noindent and $c(\sp), \bar{\kappa}(\sp)$ are defined in \eqref{PDEv}. We remind that, as shown in Appendix~\ref{app1}, a good terminal condition is $\alpha_T = 0$.

If the model coefficients are functions of the time $t$, we can use the method of \citep{Guterding2018}. The idea is to split the entire time interval $t \in [0,T]$ into $N$ subintervals of the length $\Delta t = T/N$, and approximate time-dependent model parameters by piecewise constant coefficients. Then at every interval $i$ we have $\kappa (t) = \kappa_i, \ \sigma(t) = \sigma_i, \ i=1,\ldots,N$. Accordingly, the solution $\alpha_i(t,p)$ for every interval $i$ is given by \eqref{solRicConst} where the solution from the previous time interval $\alpha_{i-1}(t,p)$ is taken as the terminal condition. Therefore, we solve \eqref{ric2} backward in time starting with $\alpha(T,p) = 0$. As shown  in \citep{Guterding2018}, this procedure is fully analytic and very fast. Indeed, on an uniform grid we need to compute function $\tan \left(C(t,p) (T-t)/2\right)$ only once while computationally this is the most expensive operation.

In the next Section we construct an inverse transform of \eqref{sol1} by using complex analysis. In a complex plane the function $C(t,p)$, as it is defined in \eqref{solRicConst}, is a multivalued function which can easily be seen if we represent it in the form
\begin{align} \label{funcC}
C(t,p) &= \sqrt{(\sqrt{p} - p_+)(\sqrt{p} - p_-)}, \qquad
p_\pm = \frac{\sigma - 2 \bar{\kappa} \rho \pm \sqrt{4 \bar{\kappa}^2 - 4 \bar{\kappa} \rho \sigma + \sigma^2}}{2 \left(1 - \rho ^2\right) \sigma }.
\end{align}
Thus, both branch (critical) points $p_\pm$ are pure real.

\subsection{Inversion}

Since functions $\sin[\xi(x- y(t))]$ form an orthonormal basis in  $[y(\tau), \infty)$ we can look for the solution $u(t, x, v)$ in the following form
\begin{equation} \label{invTr}
	P(t, x, v) = \int_0^\infty \chi(\xi, t, v) \sin[\xi(x- y(t))] d\xi,
\end{equation}
\noindent where $\chi(\xi, t, v)$ are some weights to be determined. Note, that this definition automatically respects the vanishing boundary conditions for $P(t, x, v)$. For $x = y(t)$ this is obvious. For $x \to \infty$ this can be seen looking at the final solution of a similar problem which is obtained in \citep{ItkinMuravey2020r}. We assume that the integral in \eqref{invTr} converges absolutely and uniformly $\forall x \in [y(t), \infty)$ for any $t > 0$ and $v > 0$.

Applying \eqref{GITdef} to both parts of \eqref{invTr} and integrating yields
\begin{align}
	\bar{u}(t, v, p) &= \int_{y(t)}^\infty e^{-\sqrt{p}x} \int_0^\infty \chi(\xi, t, v) \sin(\xi(x- y(t))) d\xi dx
	= e^{-\sqrt{p} y(t)} \int_0^\infty  \chi(\xi, t, v) \frac{ \xi  d\xi}{\xi^2 + p}.
\end{align}
or
\begin{equation} \label{InvEq}
	\int_0^\infty   \chi(\xi, \tau, v) \frac{ \xi  d\xi}{\xi^2 + p} = 	\bar{u}(t, v, p) e^{\sqrt{p} y(t)}.
\end{equation}
Now, similar to a standard construction of inverse operators, e.g., the inverse Laplace transform, we need an analytic continuation of the transform parameter $p$ into the  complex plane. Let us integrate both sides of \eqref{InvEq} on $p$ along the so-called keyhole contour presented in Fig.~\ref{contour}, \citep{ItkinMuravey2020r}.
\begin{figure}[!htb]
	\hspace{-0.5in}
	\captionsetup{width=0.8\linewidth}	
	\begin{center}
		\fbox{
\scalebox{0.8}{
			\begin{tikzpicture}
				\def\bigradius{4}
				\def\axisradius{4}
				\def\gammaradius{0}
				\def\omegaradius{0}
				\def\littleradius{0.4}
				\def\xii{15}
				\def\ang{10}
				\def\sh{0.07}
				\def\pole{3}
				
				\draw (-1.1*\axisradius, 0) -- (1.1*\axisradius,0)
				(\omegaradius, -1.1*\axisradius) -- (\omegaradius, 1.1*\axisradius)
				(\gammaradius, -\bigradius) -- (\gammaradius, \bigradius);
				
				\draw[red, ultra thick, decoration={ markings,
					mark=at position 0.1 with {\arrow{<}}
					,mark=at position 0.3 with {\arrow{<}}
					,mark=at position 0.4 with {\arrow{<}}
					,mark=at position 0.5 with {\arrow{<}}
					,mark=at position 0.6 with {\arrow{<}}
					,mark=at position 0.7 with {\arrow{<}}
					,mark=at position 0.85 with {\arrow{<}},
					,mark=at position 0.92 with {\arrow{<}}
					,mark=at position 0.99 with {\arrow{<}}
				},
				postaction={decorate}]
				let
				\n1 = {asin(\sh/\littleradius)},
				\n2 = {\bigradius*sin(\xii)},
				\n3 = {\littleradius*sin(\ang)},
				\n4 = {\bigradius*cos(\xii)},
				\n5 = {\bigradius*sin(\xii)},
				\n6 = {\bigradius*sin(2)},
				\n7 = {\bigradius*cos(2)}
				in
				(-\n4,-\n5) arc(-180+\xii:-2:\bigradius) -- (\littleradius, -\n6) -- (\littleradius, -\n6) arc(0:-180:\littleradius) -- (-\littleradius, -\n2)
				--  (-\pole+\sh, -\n5)  -- (-\pole+\sh, -\n3-0.35) -- (-\pole+\sh, -\n3-0.35) arc(-90:90:\littleradius)
				-- (-\pole+\sh, \n3+0.3) -- (-\pole+\sh, \n5) -- (-\littleradius, \n5) -- (-\littleradius, \n3)
				-- (-\littleradius, \n3) arc(180:7:\littleradius) --  (\bigradius, \n6-0.03) -- (\bigradius, \n6)  arc(0:180-\ang-2.8:\bigradius)
				--  (-\pole-\sh, \n5) --  (-\pole-\sh, \n3+0.3)
				--  (-\pole-\sh, \n3+0.3) arc(90:270:\littleradius) -- (-\pole-\sh, -\n5)  -- (-\n4,-\n5);

				\node at (1.2*\axisradius,-0.3){$\operatorname{Re} p$};
				\node at (0,1.2*\axisradius) {$\operatorname{Im} p$};
				\node at (0.15,-0.25){$0$};
				\node at (0.4,0.6) {$\gamma_{\varepsilon}$};
				\node at (-\pole-0.5,-0.6) {$\gamma_{r}$};
				\node at (0.,0) {$\bullet$};
				\node at (-\pole,0) {$\bullet$};
				\node at (-\pole,-0.2) {$-\xi^2$};
				\node at (-2.,3.) {$\Gamma$};
				\node at (-1.3,1.4) {$l_1 : \sqrt{p}= \iu \xi$};
				\node at (-1.55,-1.4) {$l_2 :\sqrt{p} = - \iu \xi$};
				\node at (\bigradius/2, 0.4) {$l_3$};
				\node at (\bigradius/2, -0.4) {$l_4$};
			\end{tikzpicture}
}
		}
	\end{center}
	\caption{Contour of integration of \eqref{InvEq} in a complex plane of $p$.}
	\label{contour}
\end{figure}
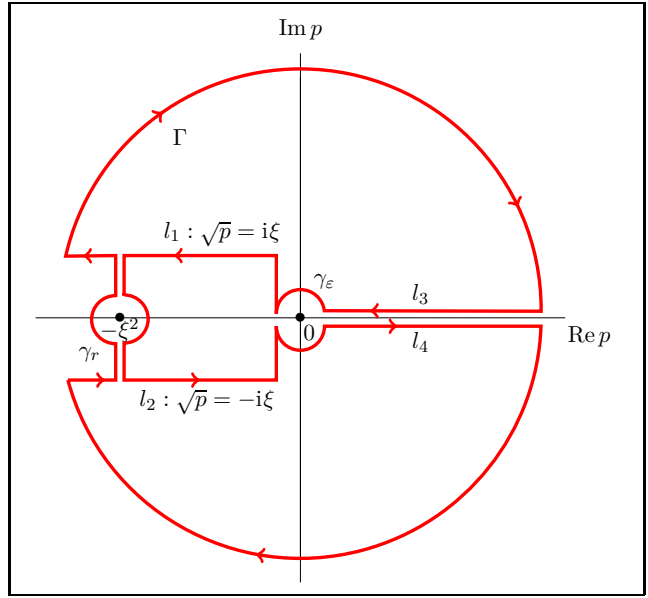
In more detail, this contour can be described as follows. It starts with a big symmetric arc $\Gamma$ around the origin with the radius $R$; extending to two horizontal line segments $l_3, l_4$ (a cut around the line $\operatorname{Im} p = 0, \operatorname{Re} p > 0$); connecting to two small semi-circles $\gamma_\varepsilon$ around the origin with the radius $\varepsilon \ll 1$; then extending to two vertical line segments up to points  $\operatorname{Im}(\sqrt{p}) = \pm \xi$; then again two horizontal parallel line segments  $l_1, l_2$ at $\operatorname{Im} \sqrt{p} = \pm \xi$, which end points are connected to the arc $\Gamma$ with a cut at $\operatorname{Im} p = -\xi^2$ (it consists of two vertical line segments and two semi-circles $\gamma_r$ with the radius $\varepsilon$), such that the whole contour is continuous.

Using a standard technique, we take a limit $\varepsilon \to 0, R \to \infty$, so in this limit the contour takes the form as
depicted in Fig,~\ref{Fig2}. It has a horizontal cut along the positive real line with point $p=0$ excluded from the area inside the contour; another vertical cut at $\Ree(p) = -\xi^2$ with the point $p = -\xi^2$ lying inside the contour; and a branch cut $l_1, l_2$ of the multivalued function $\sqrt{p}$ at $p = -\xi^2$. Also, in this limit $l_7 \to 0, l_8 \to 0$, but in Fig.~\ref{Fig2} we left them as it is for a better readability.

\begin{figure}[!htb]
	\hspace{-0.5in}
	\captionsetup{width=0.8\linewidth}	
	\begin{center}
		\fbox{
\scalebox{0.7}{
			\begin{tikzpicture}
				\def\bigradius{4}
				\def\axisradius{4}
				\def\gammaradius{0}
				\def\omegaradius{0}
				\def\littleradius{0.4}
				\def\xii{15}
				\def\ang{10}
				\def\sh{0.07}
				\def\pole{3}
				\def\vert{1.1}
				
				\draw (-1.1*\axisradius, 0) -- (1.1*\axisradius,0)
				(\omegaradius, -1.1*\axisradius) -- (\omegaradius, 1.1*\axisradius)
				(\gammaradius, -\bigradius) -- (\gammaradius, \bigradius);
				
				\draw[red, ultra thick, decoration={ markings,
					mark=at position 0.1 with {\arrow{<}}
					,mark=at position 0.2 with {\arrow{>}}
					,mark=at position 0.52 with {\arrow{>}}
					,mark=at position 0.9 with {\arrow{>}}
				},
				postaction={decorate}]
				(1.1*\bigradius, 0) -- (0,0) -- (0,\vert) -- (-1.1*\bigradius,\vert);
				
				\draw[red, ultra thick, decoration={ markings,
					mark=at position 0.1 with {\arrow{<}}
					,mark=at position 0.9 with {\arrow{<}}
				},
				postaction={decorate}]
				(0,0) -- (0,-\vert) -- (-1.1*\bigradius,-\vert);
				
				\draw[red, ultra thick, decoration={ markings,
					mark=at position 0.2 with {\arrow{>}}
					,mark=at position 0.8 with {\arrow{<}}
				},
				postaction={decorate}]
				(-\pole,\vert) -- (-\pole,-\vert);

				\node at (1.2*\axisradius,-0.3){$\operatorname{Re} p$};
				\node at (0,1.2*\axisradius) {$\operatorname{Im} p$};
				\node at (-0.2,-0.3){$0$};
				\node at (0.,0) {$\bullet$};
				\node at (-\pole-0.4,-0.3) {$-\xi^2$};
				\node at (-1.1,1.4) {$l_1$};
				\node at (-1.555,-1.4) {$l_2$};
				\node at (\bigradius/2, 0.4) {$l_3$};
				\node at (\bigradius/2, -0.4) {$l_4$};
				\node at (-\pole-0.4,0.5*\vert) {$l_5$};
				\node at (-\pole+0.4,-0.5*\vert) {$l_6$};
				\node at (0.4,-0.5*\vert) {$l_7$};
				\node at (0.4,0.5*\vert) {$l_8$};
				\node at (-\pole, 2*\vert) {$\gamma$};
			\end{tikzpicture}
}
		}
	\end{center}
	\caption{Contour of integration $\gamma$ of \eqref{InvEq} in a complex plane of $p$ at $\varepsilon \to 0, R \to \infty$.}
	\label{Fig2}
\end{figure}
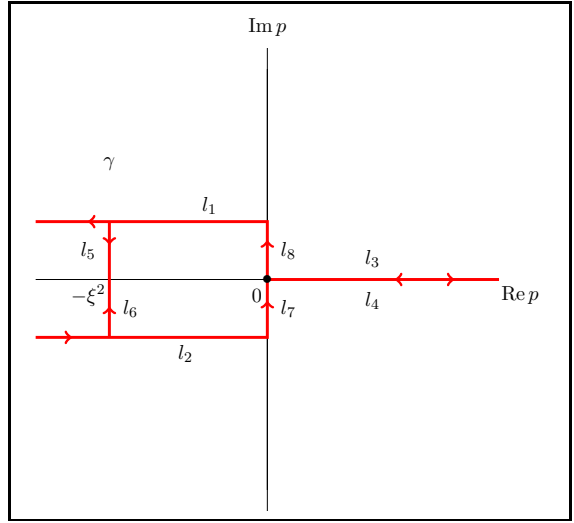

Further, let us emphasize that: i) function  $\bar{u}(t,v,p)$ in \eqref{sol1} is a function of $\alpha$ (and via this dependence function $\beta$ is also a function of $\alpha$); ii)  as we mentioned at the end of Section~\ref{secLMVF}, $\alpha$ is a multivalued function (since $C(t,p)$ is a multivalued function). Due to this, the contour in Fig.~\ref{contour} has to be updated with the corresponding branches for $C(t,p)$. However, since both critical points  $p_\pm$ in \eqref{funcC} are real, this will add two branch cuts to the contour along the real line. It can be checked that these cuts don't contribute to the contour integral under consideration, hence, we don't show them in Fig.~\ref{contour} to make the picture better readable.

Now we are ready to compute the integrals in \eqref{InvEq}. That one in the LHS is regular everywhere inside this contour except the single pole $p = -\xi^2$. By the residue theorem, we obtain
\begin{equation} \label{Int1}
	\oint_{\gamma}  \left( \int_0^\infty  \chi(\xi, t, v) \frac{\xi d\xi}{\xi^2 + p} \right)dp = -2\pi \iu \int_0^\infty \xi \chi(\xi, t, v) d\xi.
\end{equation}

The integral in the RHS of \eqref{InvEq} doesn't have any singularity inside the contour $\gamma$, however, it has several cuts. As can be easily checked, the integrals along the segments $l_3$ and $l_4$ cancel out, as well as those along $l_7$ and $l_8$, and those along $l_5$ and $l_6$.  The integral along the contour $\Gamma$ tends to zero if $R \to \infty $ due to Jordan's lemma. Hence, the only remaining integrals are those along the horizontal semi-infinite lines $l_1$ and $l_2$. They read
\begin{align} \label{Int2}
\int_{l_1} & \bar{u}(t, v, p) e^{\sqrt{p} y(t)}dp = 2 \int_0^\infty \xi \bar{u}(t, v, e^{\iu \pi} \xi^2) e^{\iu \xi y(t)}d\xi,  \\
\int_{l_2}& \bar{u}(t, v, p) e^{\sqrt{p} y(t)}dp = -2 \int_0^\infty \xi \bar{u}(t, v, e^{-\iu \pi} \xi^2) e^{-\iu \xi y(t)}d\xi. \nonumber
\end{align}
Combining \eqref{Int1} with \eqref{Int2} yields
\begin{equation} \label{chi1}
\chi(\xi, t, v) = \frac{1}{\pi \iu} \left[\bar{u}(t, v, e^{-\iu \pi} \xi^2) e^{-\iu \xi y(t)} - \bar{u}(t, v, e^{\iu \pi} \xi^2) e^{\iu \xi y(t)} \right].
\end{equation}
%
Therefore,
\begin{equation} \label{solu}
P(t, x, v) = \frac{1}{\pi \iu} \int_0^\infty \sin\left[\xi (x - y(t))\right] \left[\bar{u}(t, v, e^{-\iu \pi} \xi^2) e^{-\iu \xi y(t)} - \bar{u}(t, v, e^{\iu \pi} \xi^2) e^{\iu \xi y(t)} \right] d\xi.
\end{equation}
Here the argument $p = e^{\pm \iu \pi} \xi^2$ implies $\sqrt{p} = \pm \iu \xi$.  It can be checked (at least, numerically) that
\begin{equation*}
\Ree\left[\bar{u}(t, v, e^{-\iu \pi} \xi^2)\right] = \Ree\left[\bar{u}(t, v, e^{\iu \pi} \xi^2)\right],
\end{equation*}
\noindent and so
\begin{equation*}
\Ree\left[\bar{u}(t, v, e^{-\iu \pi} \xi^2) e^{-\iu \xi y(t)}\right] = \Ree\left[\bar{u}(t, v, e^{\iu \pi} \xi^2) e^{\iu \xi y(t)} \right].
\end{equation*}
Thus, $\chi(\xi, t, v)$ is real. Accordingly, by simple algebra $\bar{u}(t, v, e^{\iu \pi} \xi^2) = (\bar{u}(t, v, e^{-\iu \pi} \xi^2))^*$. Therefore,
\begin{equation} \label{chi}
\chi(\xi, t, v) = -\frac{2}{\pi} \left[\bar{u}_I(t, v, e^{-\iu \pi} \xi^2) \cos(\xi y(t)) + \bar{u}_R(t, v, e^{-\iu \pi} \xi^2) \sin(\xi y(t)) \right],
\end{equation}
\noindent where sub-indices $R, I$ denote the real and imaginary parts.

Substitution of the explicit representation of $\baru(t,v,p)$ in \eqref{sol1} into \eqref{solu} yields the final representation of the solution. This result can be summarized as the following Proposition\footnote{Below in several places we use a controversial notation where a function argument $p$ is replaced with $\sqrt{p}$. This, however,  allows writing many formulae in a general way, further having in mind that they should be used with $\sqrt{p} = -\iu \xi$ and $\sqrt{p} = \iu \xi$. In both cases this implies that $p = - \xi^2$.}

\begin{proposition} \label{prop1}
Let us consider a time-dependent Heston stochastic volatility model defined in \eqref{heston} with the additional condition in \eqref{cond1} that
\begin{equation}
\frac{\kappa(t) \theta(t)}{\sigma^2(t)} = \frac{m}{2},
\end{equation}
\noindent where $m \in [0,\infty)$ is some constant.  Also, let us consider a Down-and-Out barrier Put option written of the underlying which follows  \eqref{heston}. Let the lower barrier $L(t)$ be time-dependent as defined in \eqref{bcH}, and let $y(t) = \log(L(t)/K)$. Given the values of the log-spot $x = \log(S/K)$ and the instantaneous variance $v$ at the initial moment of time $t$, the price of this option $P(t,x,v)$ is given by
\begin{align} \label{finSol}
P(t, x, v) &= -\frac{2}{\pi}\int_0^\infty \sin\left[\xi (x - y(t))\right] \Big\{ \cos(\xi  y(t)) \Imm[P_s(t,v,-\iu \xi)]
+ \sin(\xi  y(t)) \Ree[P_s(t,v,-\iu \xi)] \Big\} d\xi, \nonumber \\
&= -\frac{1}{\pi}\int_0^\infty |P_s(t,v,-\iu \xi)| \left\{\cos[\phi - \xi(x-2 y(t))] - \cos[\phi + x \xi] \right\} d\xi, \\
P_s(t,v,-\iu \xi) &= P_{1}(t,v,-\iu \xi)  +  P_2(t,v,-\iu \xi), \quad |P_s(t,v,-\iu \xi)|^2 = [\Ree P_s(t,v,-\iu \xi)]^2 + [\Imm P_s(t,v,-\iu \xi)]^2,  \nonumber \\
P_1(t,v,\sqrt{p}) &= K \left[ \frac{ e^{-y(T) \sqrt{p}} - 1}{\sqrt{p}} - \frac{e^{-y(T) \left(\sqrt{p}-1\right)} - 1}{\sqrt{p}-1} \right] e^{\beta(t,p) + \gamma(t,p) v}, \quad \phi = \arg(P_s(t,v,-\iu \xi)), \nonumber \\
P_2(t,v,\sqrt{p}) &= \frac{1}{2}  \int_t^T d s \int_0^\infty d v' \Bigg[  \Phi\left(s, v' \right) \sqrt{v'} g(s,p) \nonumber \\
&\times G\left(\frac{1}{4}\int_t^s g^2(\zeta, p) \sigma^2(\zeta) d \zeta, \, g(t,p) \sqrt{v},\, g(s, p) \sqrt{v'} \right) e^{ -y(s) \sqrt{p}  +\alpha(t,p) v + \beta(t,p) - \left[\beta(s, p) + v'\alpha(s,p)\right]} \Bigg], \nonumber \\
B_1 &= \frac{\alpha(T,p)}{g^2(T,p)}, \quad \gamma(t,p) = \alpha(t,p) - \frac{\alpha(T,p)}{1 + 2 B_1 \tau} \frac{g^2(t,p)}{g^2(T,p)}. \nonumber
\end{align}
Once the function $\Phi\left(t, v \right)$ is known (which is a half of  the gradient (in $x$) of the solution $P(t,v,x)$ at the boundary $x=y(t)$), the solution of this pricing problem is obtained via \eqref{finSol} by computing the integrals in the RHS.
\hfill $\square$
\end{proposition}

Note, that alternatively the inversion formula \eqref{solu} (and, accordingly, the option Put price in \eqref{finSol}) could be  derived via Fourier-sine transform, see Appendix ~\ref{appSinFT}.

Both integrals in the RHS of the definition of $P_1, P_2$ in \eqref{finSol} should be well-behaved at $v \to \infty$. This can be achieved by choosing an appropriate terminal condition for $\alpha(t,p)$ in \eqref{ric2} and is discussed in Appendix~\ref{app1}. It is shown there that a good terminal condition could be $\alpha(T,p) = 0$, so $B_1 = 0$.

Also, it can be directly checked that $\Ree(P_1(t,v,\iu \xi)) = \Ree(P_1(t,v,-\iu \xi))$ and, hence, the difference of $P_1$ is pure imaginary. Therefore, $[P_1(t,v,-\iu \xi)) - P_1(t,v,-\iu \xi)]/\iu$ is real. Same should be true for the difference of $P_2$, however, this can be verified only numerically.

Similar to the one-dimensional case described in detail in \citep{ItkinLiptonMuraveyBook}, the function $\Phi\left(t, v \right)$ solves a linear mixed Volterra-Fredholm (LMVF) integral equation of the second kind. It can be obtained by differentiating both sides of \eqref{finSol} with respect to $x$  and setting $x = y(t)$. Assuming that $\Phi\left(t, v \right) \in \mathbb{R}$, this yields
\begin{align} \label{volt}
f(t,v) &= \Phi(t,v) + \frac{1}{2\pi} \int_t^T d s \int_0^\infty d v'  \Phi\left(s, v' \right) {\cal K}(s, v', t,v), \\
f(t,v) &= -\frac{1}{\pi}\int_0^\infty \xi
\Big\{ \cos(\xi  y(t)) \Imm[P_1(t,v,-\iu \xi)] + \sin(\xi  y(t)) \Ree[P_1(t,v,-\iu \xi)] \Big\} d\xi, \nonumber
\end{align}
\noindent and ${\cal K}(s, v', t,v)$ is the kernel of this LMVF integral equation which reads
\begin{align} \label{kernel}
{\cal K}(s, v', t,v) &= \int_0^\infty \xi \Big\{ \cos(\xi  y(t)) \Imm[{\mathfrak K}(s, v', t,v, -\iu \xi)] + \sin(\xi  y(t)) \Ree[{\mathfrak K}(s, v', t,v, -\iu \xi)] \Big\} d \xi, \\
{\mathfrak K }(s, v', t,v, \sqrt{p}) &= \sqrt{v'} g(s,p)  G\left(\frac{1}{4}\int_t^s g^2(\zeta, p) \sigma^2(\zeta) d \zeta, \, g(t,p) \sqrt{v},\, g(s, p) \sqrt{v'} \right) \nonumber \\
&\times \exp\left[ -y(s)\sqrt{p} +\alpha(t,p) v + \beta(t,p) - \left(\beta(s, p) + v'\alpha(s,p)\right) \right]. \nonumber
\end{align}

Thus, we have managed to generalize the GIT method originally proposed in \citep{CarrItkin2020jd} for solving one-dimensional financial problems with moving barriers and further developed in a series of papers and \citep{ItkinLiptonMuraveyBook} to solving similar problems for the models with stochastic volatility. Note, that various advantages of the GIT method as applied to one-dimensional problems are reported in the above cited papers. However, here, for the two-dimensional (2D) problem  the drawback is that, in contrast to the one-dimensional counterparts, the integral on $\xi$ cannot be taken analytically. Therefore, our 2D LMVF equation instead of a closed form kernel has the one which is expressed via an integral in \eqref{kernel}. At the first glance this should significantly slow down computation of the gradient $\Phi(t,v)$. However, as shown in Section~\ref{solVolt}, the method of radial basis functions (RBF) being used for solving \eqref{volt} allows reduction of the three-dimensional integral to a 2D one in variables $(t,\xi)$. Therefore, our approach seems to be a natural extension of the GIT method to the 2D case while preserving all nice features of the method.

\section{Solution of the LMVF equation} \label{solVolt}

The LMVF equation in \eqref{volt} can be solved by using various numerical methods. Here we utilize the Radial Basis Functions method as this was proposed in \citep{Assari2019,Zhang2014, ItkinMuraveySABR}  (see also references therein).
A short description of the RBF method is given in the next section (a more detailed discussion can be found in \citep{ItkinMuraveySABR}).

\subsection{Basics of the RBF method}

The main idea of the RBF method is as follows. Interpolation of functions by using RBFs is known to be very efficient when solving various problems of intermediate ($10 > d > 3$) dimensionality including those in mathematical finance, see, e.g.,  \cite{YCHon3, Fasshauer2, Pettersson, FornbergFlyer2015} and also references in \citep{Assari2019}. It converges exponentially when increasing the number of nodes and is meshless. The latter allows obtaining a high-resolution scheme using just a few discretization nodes.

To make the further exposition transparent,  let us provide some definitions along the lines of \citep{Assari2019,ItkinMuraveySABR}. A function $\Theta: \mathbb{R}^{d} \rightarrow \mathbb{R}$ is called to be radial if there exists a univariate function $\phi:[0, \infty) \rightarrow \mathbb{R}$ such that
\begin{equation}
\Theta(\mathbf{x})=\phi(r)
\end{equation}
\noindent  where $r=\|\mathbf{x}\|$ and $\|\cdot\|$ is some norm in $\mathbb{R}^{d}$. In this paper we consider just the Euclidean norm. Let $\chi=\left\{\mathbf{x}_{1}, \ldots, \mathbf{x}_{N}\right\}$ be a set of scattered points selected in the domain $\Omega \subset \mathbb{R}^{d}$. A function $u(\mathbf{x})$ at an arbitrary point $\mathbf{x} \in \Omega$ can be approximated by using the global radial function $\phi(\|\mathbf{x}\|)$ via a linear combination
\begin{equation} \label{GlInterp}
u(\mathbf{x}) \approx \mathcal{G}_{N} u(\mathbf{x}) = \sum_{i=1}^{N} c_{i} \phi\left(\left\|\mathbf{x}-\mathbf{x}_{i}\right\|\right), \quad \mathbf{x} \in \Omega,
\end{equation}
\noindent where the coefficients $\left\{c_{1}, \ldots, c_{N}\right\}$ are determined by the interpolation conditions
\begin{equation} \label{VoltRBF}
\mathcal{G}_{N} u\left(\mathbf{x}_{i}\right)=u\left(\mathbf{x}_{i}\right)=u_{i}, \quad i=1, \ldots, N.
\end{equation}

In the literature various choices of the RBFs exist. Among others, let us mention the Gaussian RBF
\begin{equation} \label{gauss}
\Theta(\mathbf{x}) = e^{-\varepsilon \|\mathbf{x}\|^{2}},
\end{equation}
\noindent where $\varepsilon>0$ is the shape parameter. This function is strictly positive-definite in $\mathbb{R}^{d}$ and, therefore, the expansion in \eqref{GlInterp} is non-singular.

As mentioned, the RBF methods belong to the class of meshfree methods. That means that no regular grid in $\tau_j, z_l$ is required to run it (in contrast, e.g. to the FD method).  Therefore, taking a 2D set of collocation nodes $\{\tau_{j(l)}, z_l\}, \ l = 1,...,\bar{l}, \ j(l)=1,...,\bar{j}(l)$ one can substitute them into \eqref{VoltRBF} and get a system of linear equations for the coefficients $c_{j,l}$. For instance, in case of the regular grid with $N_l$ nodes in $z_l$ and $N_j$ nodes in $\tau_j$ we have $N_j N_l$ unknown coefficients which solve a system of $N_j N_l$ linear equations. The matrix of this system is dense and, therefore, complexity of solving this system by using the direct solver is $O(N^3_j N^3_l)$. Obviously, this result is not satisfactory from computational point of view. Iterative solvers can improve this especially when a suitable preconditioner can be constructed.

The global approach can be significantly improved in a several ways. For instance, a local version of the method estimates the solution using only the discrete collocation nodes and locally supported RBFs constructed on a small set of nodes instead of all points over the analyzed domain. This approach enables a significant reduction in the number of non-zero elements that remain in the coefficient matrix, hence, lowering the computational intensity required for solving the system. As shown in \citep{Assari2019} the complexity of such a scheme (for the 2D problem) drops down to $O(N_1 N_2 m_1 m_2)$ where $m_i, \ i=1,2$ is the number of the corresponding integration nodes (those that are used in a quadrature scheme when approximating the integrals in the LMVF equation, e.g., by using the Gauss–Legendre integration rule on the local influence domain). Also,  in comparison with the globally supported RBF for solving integral equations, the method of \citep{Assari2019} is stable and uses much less computer memory.

An alternative way of making the global RBF better is using a "better basis" for RBF interpolation. Indeed, it is well-known, e.g., \citep{McCourt2012}, that the global Gaussian RBF method leads to a notoriously ill-conditioned interpolation matrix whenever $\varepsilon$ is small and the set of the basis functions in \eqref{gauss} becomes numerically linearly dependent on $\mathbb{R}^d$. This leads to severe numerical instabilities and limits the practical use of Gaussians — even though one can approximate a function with the Gaussian kernel with spectral approximation rates. On the other hand, small $\varepsilon$ provide better accuracy, and so have to be considered as an option for an accurate pricing. It is also known that if $\varepsilon$ is kept fixed, convergence stagnates even if $N_j, N_l$ grow, and if $N_j, N_l$ are fixed, the error blows up with the decrease of $\varepsilon$. Therefore, to obtain well-conditioned (and therefore numerically stable) interpolation among others, let us mention the RBF-QR method, \citep{Fornberg2, Larsson2, Larsson7}.

Having this in mind, in this paper we, however, use just a simple version of the global Gaussian RBF method. This is done for two reasons. First, here we want to illustrate that the proposed method combined even with Gaussian RBFs provides reasonable option prices. Second, even with this global method the speed of computations is better than that of the FD method. Third, further improvement of the method is subject of a separate research which will be presented elsewhere.

As follows from \eqref{volt}, the final linear system of algebraic equations for finding coefficients of the RBF interpolation has the form
\begin{equation} \label{linSys}
\|A\| |c| = K |f|,
\end{equation}
\noindent where $\|A\|$ is the matrix of the discretized RHS of \eqref{volt}, $|c|$ is the vector of unknown coefficients and $|f|$ is the vector of the discretized LHS of \eqref{volt} with multiplier $K$ dropped away. As follows from \eqref{chi1} and representation of $\bar{u}(t,v,p)$ in \eqref{sol1}, the matrix $\|A\|$ doesn't depend on $K$, and the vector $|f|$ weakly depends on $K$. Therefore, to find prices  of options with same maturity $T$ and various strikes $K$ this system can be solved just ones since this is a linear system with the same matrix and multiple RHS. This can be efficiently done by using any parallel architecture. Hence, all coefficients $c(K_i), \ i=1,...$ can be computed by solving only a single LMVF equation. Then, the prices for all strikes can be found from \eqref{finSol}, where the dependence on $K$ comes from $x = \log(S/K), P_1$ and $\Phi(s,v')$. Those integrals (for various K) can be computed in one sweep using software which supports vectorization, e.g. MATLAB. Thus,  our method allows almost simultaneous computation of the barrier  option prices for all strikes.

\subsection{Numerical scheme}

To employ the RBF interpolation described in the previous section for solving the LMVF equation we need a set of $N$ collocation points $\{ (t_1, v_1),\ldots,(t_1,v_{N_v}),\ldots,(t_{N_t}, v_1),\ldots,(t_{N_t},x_{N_v})\}$ in the 2D space $\calB = [0,T] \times [0,\infty)$. Then the unknown solution $\Phi(t,v)$ can be approximated using the RBF method as
\begin{equation} \label{approxRBF}
\Phi(t,v) \approx \sum_{i=1}^{N_t N_v} c_{i} \phi\left(\left\|\mathbf{x}-\mathbf{x}_{i}\right\|\right)
= \sum_{k=1}^{N_t} \sum_{l=1}^{N_v} c_{kl} e^{-\varepsilon [(v_l - v)^2 + (t_k - t)^2]}.
\end{equation}

Next this representation should be substituted into \eqref{volt} to obtain a system of equations for the coefficients
$c_{jl}$. This, however, requires  discretization of  the LMVF integrals by using some quadrature rules.
Note, that we have three integrals in the RHS of \eqref{volt}. Therefore, after substituting \eqref{approxRBF} into \eqref{volt} all triple integrals in the RHS of \eqref{volt} acquire the form
\begin{align} \label{Iki}
I_{kl}(t,v) &= \int_t^T d s \int_0^\infty d v'  \Phi\left(s, v' \right) e^{-\varepsilon [(v' - v_l)^2 + (s - t_k)^2]}
 \int_0^\infty  d \xi \, \xi \sqrt{v'}  e^{-y(s)\sqrt{p} +\alpha(t,p) v + \beta(t,p)- \beta(s, p) - \alpha(s,p) v'} \nonumber \\
&\times g(s,p) \frac{\sqrt{g(t,p) g(s, p) \sqrt{v v'}}}{\tau(t,p) - \tau(s,p)} \left(\frac{g(s, p) \sqrt{v'}}{g(t,p) \sqrt{v}}\right)^b e^{- \frac{g^2(t,p) v + g^2(s, p) v'}{2 (\tau(t,p) - \tau(s,p))}}
I_{b-\frac{1}{2}} \left(\frac{g(t,p) g(s,p)\sqrt{v v'}}{(\tau(t,p) - \tau(s,p))} \right) \\
&=   v^{1/4 - b/2} \int_t^T d s \,  e^{-\varepsilon (s - t_k)^2} \int_0^\infty d\xi \, \xi e^{- y(s)\sqrt{p} + \beta(t,p) - \beta(s, p) + \alpha(t,p) v}  \left(\frac{g(s, p)}{g(t,p)}\right)^{b-1/2} \frac{e^{- \frac{g^2(t,p) v}{2 (\tau(t,p) - \tau(s,p))}}}{\tau(t,p) - \tau(s,p)}  \nonumber \\
&\times g^2(s, p) \int_0^\infty dv' (v')^{3/4 + b/2}  e^{-\varepsilon (v' - v_l)^2 - \frac{g^2(s,p) v'}{2 (\tau(t,p) - \tau(s,p))} - \alpha(s,p) v'} I_{b-\frac{1}{2}} \left(\frac{g(t,p) g(s,p)\sqrt{v v'}}{\tau(t,p) - \tau(s,p)} \right). \nonumber
\end{align}
Unfortunately, to the best of our knowledge the last integral in $v'$ (let us denote is as $\calJ$) cannot be taken analytically.

It turns out, however, that by using a special trick we are able to reduce this method to computation of just two integrals. For doing so, let us introduce a new function $\phi(s,v')$ which, strictly speaking, is not an RBF but behaves like a Gaussian RBF.  Hence, instead of the Gaussian RBF
\begin{equation} \label{gauss1}
\Theta_{kl}(s,v') = e^{-\varepsilon [(v'-v_l)^2 + (s-t_k)^2]},
\end{equation}
\noindent we propose another function
\begin{equation} \label{newGauss}
\bar{\Theta}_{kl}(s,\nu') = \left(\frac{\nu'}{\nu_l}\right)^{2 \varepsilon \nu_l^{2}}  e^{-\varepsilon \left[\nu'^2 - \nu_l^2 + (s - t_k)^2 \right]} +8 \delta(\nu'),
\end{equation}
\noindent where a new variable $\nu = \sqrt{v}$ has been introduced, and $\delta(\nu)$ is the Dirac delta function, \citep{as64}, which is necessary in this definition to have $\bar{\Theta}_{kl}(s,\nu')$ to be positive-definite. The latter condition is required because if the basis function is strictly positive definite then the associated interpolation matrix $\|\Theta\|$ is positive definite and, hence,  non-singular. Therefore, the interpolation problem will be well-posed and has a unique solution, \citep{Fasshauer}. Proof of the positive-definiteness of $\bar{\Theta}(s,\nu')$ is given in Appendix~\ref{appPD}.

The function $\bar{\Theta}_{kl}(s,\nu')$ behaves similar to another Gaussian RBF
\begin{equation} \label{gaussSqrt}
G_{kl}(s,\sqrt{v'}) = e^{-\varepsilon [(\sqrt{v'}-\sqrt{v_l})^2 + (s-t_k)^2]}.
\end{equation}
Since $v' \in [0,\infty]$ the function $G_{kl}(s,\sqrt{v'})$ is a good Gaussian RBF.

It can be checked that $\bar{\Theta}_{kl}(s,\nu')$ has its maximum in $\nu'$ at $\nu' = \nu_l$ and $\bar{\Theta}_{kl} (t_k,\nu_l) = 1$ while aside of $(t_k,\nu_l)$ it rapidly vanishes. Also, when $\nu_l \to 0$ we have the correct limit
$\bar{\Theta}_{kl}(s,\nu') \to e^{-\epsilon  \left[\nu'^2 + (s-t_k)^2 \right]}$. Comparative behavior of $\bar{\Theta}_{kl}(s,\nu')$ and $G_{kl}(s,\nu')$ is shown in Fig.~\ref{PhiComp}.
\begin{figure}[!htb]
\begin{center}
\hspace*{-0.3in}
\subfloat[]{\includegraphics[width=0.45\textwidth]{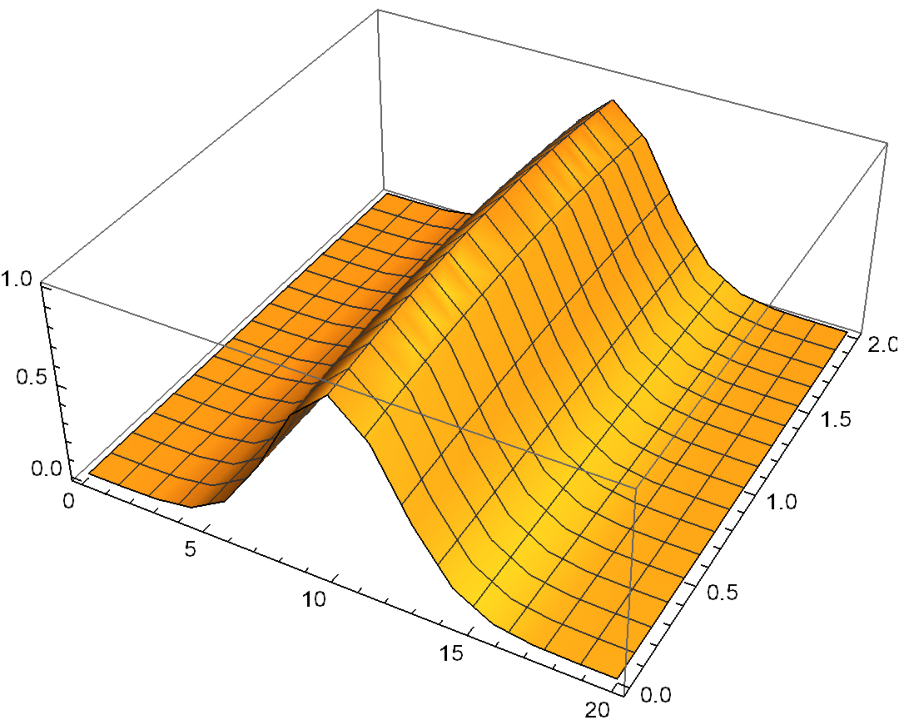}}
\subfloat[]{\includegraphics[width=0.45\textwidth]{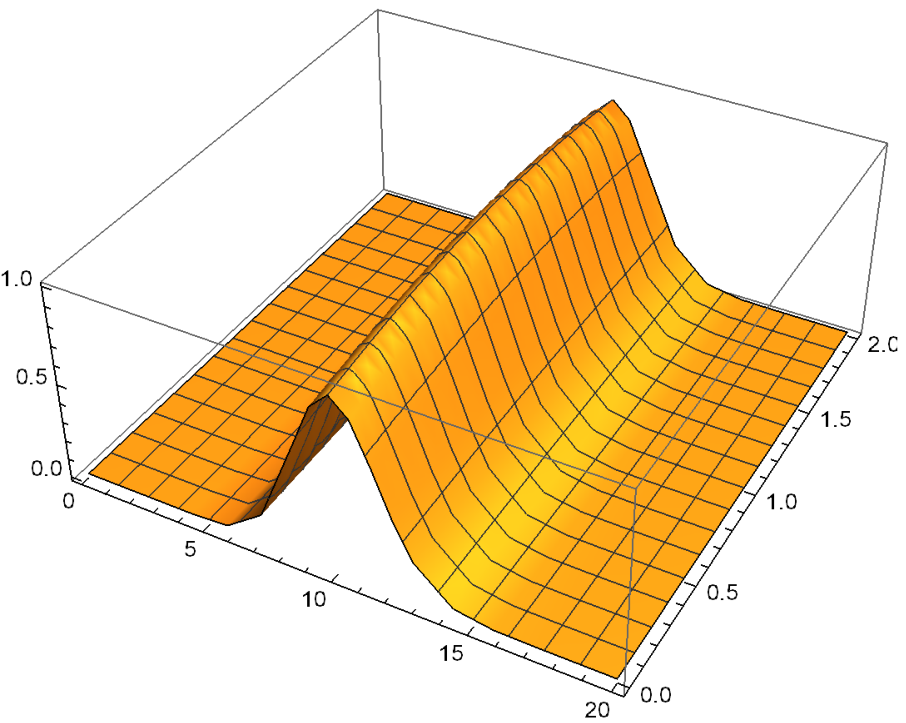}}
\end{center}
\caption{The behavior of functions a) $G(t,\nu)$ and b) $\bar{\Theta}(t,\nu)$ -  at $\varepsilon =  0.1, \ \nu_l = 10, \ t_k = 1$.}
\label{PhiComp}
\end{figure}

Using $\bar{\Theta}_{kl}(s,\nu')$ in the RBF method in the same way as we previously did it for $\Theta_{kl}(s,v')$ gives rise to the following  transformed expression for $I_{kl}(t,v)$
\begin{align} \label{IkiNew}
I_{kl}(t,v) &=   v^{1/4 - b/2} \int_t^T d s \,  e^{-\varepsilon (s - t_k)^2} \int_0^\infty d\xi \, \xi e^{- y(s) \sqrt{p} + \beta(t,p) - \beta(s, p) + \alpha(t,p) v}  \left(\frac{g(s, p)}{g(t,p)}\right)^{b-1/2} \\
&\times \frac{e^{- \frac{g^2(t,p) v}{2 (\tau(t,p) - \tau(s,p))}}}{\tau(t,p) - \tau(s,p)} g^2(s, p)
\int_0^\infty dv' (v')^{3/4 + b/2}  \left(\frac{\nu'}{\nu_l}\right)^{2 \varepsilon \nu_l^{2}}  e^{-\varepsilon (\nu'^2 - \nu_l^2)} e^{- \frac{g^2(s,p) v'}{2 (\tau(t,p) - \tau(s,p))} - \alpha(s,p) v'} \nonumber \\
&\times I_{b-\frac{1}{2}} \left(\frac{g(t,p) g(s,p)\sqrt{v v'}}{\tau(t,p) - \tau(s,p)} \right). \nonumber
\end{align}
Now, having in mind that $v' = \nu^{'2}$, the last integral (let us denote it as $\cal J$) can be computed in closed form since
\begin{align} \label{Jint}
\calJ &= a_0 \int_0^\infty \nu^{a_1}  e^{-a_2 \nu^2} I_{b-\frac{1}{2}}(a_3 \nu)d\nu  = a_0 2^{-b-\frac{1}{2}} a_3^{b-\frac{1}{2}} a_2^{-a_4} \frac{\Gamma \left(a_4\right)}{\Gamma\left(b+\frac{1}{2}\right)} M \left(a_4; b+\frac{1}{2}; \frac{a_3^2}{4 a_2}\right),
\end{align}
\noindent and in our case
\begin{alignat}{2} \label{aCoeff}
a_0 &= 2 (e/v_l)^{\varepsilon v_l}, &\qquad a_2 &= \varepsilon + \frac{g^2(s,p)}{2 (\tau(t,p) - \tau(s,p))} + \alpha(s,p), \\
a_1 &= \frac{5}{2} + b + 2 \varepsilon v_l, &\qquad a_3 &= \frac{g(t,p) g(s,p)\sqrt{v}}{\tau(t,p) - \tau(s,p)},
\qquad a_4 = \frac{1}{4} (2 a_1 + 2 b + 1) = b + \frac{3}{2} + \varepsilon v_l. \nonumber
\end{alignat}
Here $\Gamma(x)$ is the gamma function, and $M(a,b,z)$ is the Kummer confluent hypergeometric function, \citep{as64}.

Note that the Delta function in the definition of $\bar{\Theta}_{kl}(t,\nu)$ doesn't bring any problem with solving the LMVF equation. That is because first, it doesn't contribute to the $I_{kl}(t,v)$ since the corresponding integrand is proportional to $\nu = \sqrt{v'}$, hence the integral vanishes at $\nu' = 0$. Second, a set of collocation points can be chosen to exclude the point $v_l = 0$ and instead replace it with the point $v = \epsilon \ll 1$ which doesn't influence the quality of approximation. Then the term with the Delta function in the representation of $\Theta_{kl}(t,\nu)$ vanishes as well.

Thus, we have managed to reduce a triple integral to a double one using the new basis function $\bar{\Theta}(t,\nu)$. Let us underline that the RHS of $I_{kl}(t,v)$ depends on $v$ explicitly, and we don't need any numerical approximation of the integral on $\nu'$. Instead, we need to use some quadrature rule to compute the integral on $\xi$ given the values of $s$ and $v'$. Overall, the solution of the LMVF equation needs computation of a 2D integral, i.e. the complexity of getting this solution fits the scope of the standard RBF method for 2D integral equations, like in \citep{Assari2019}.

It turns out, that the RHS of \eqref{Jint} can be further simplified. Indeed, practical values of $v_l$ lie in the region, e.g., $0 < v_l < 2$ while a typical value of $\varepsilon$ is $\varepsilon \approx O(0.1)$. Also, from Proposition~\ref{prop2} it follows that $b > 0$ if the model parameters satisfy the Feller condition, and $b > -1/2$ otherwise. Therefore, in the zero-order  approximation in $\varepsilon$ we can set $a_4 \approx b + 3/2$. Then, the following identity holds, \citep{NIST:DLMF}
\begin{equation} \label{KumApp}
M \left(b+\frac{3}{2}; b+\frac{1}{2}; x\right) = \left(1 + \frac{2x}{1+2b}\right)e^x.
\end{equation}
Our numerical experiments show that this  approximation works well and produces a small error in the final option price. Accordingly, we get
\begin{align} \label{Jint1}
\calJ &= 2^{1/2-b} C_1 \left(\frac{e}{v_l}\right)^{\varepsilon v_l}  \dfrac{a_3^{b-1/2}}{a_2^{b+3/2}}
\left(b + \frac{1}{2} + a_5\right)e^{a_5}, \qquad a_5 = \frac{a_3^2}{4 a_2}, \quad C_1 = e^{-\varepsilon v} v^{\varepsilon v_l}.
\end{align}
The factor $C_1$ in the zero-order approximation in $\varepsilon \ll 1$ should be equal to one. However, we include it here to have the correct limiting value of $\calJ$ when $ s-t \ll 1$ (see below).

When numerically computing the integral in $s$ in \eqref{IkiNew}, one has to take into account that at $s=t$ the Green function in \eqref{kernel} becomes the Dirac delta function $\delta(v' - v)$. Therefore, in this case
\begin{align} \label{kernelDelta}
\int_t^T d s \int_0^\infty & d v'  \Phi\left(s, v' \right) \xi \left[{\mathfrak K}(s, v', t,v, -\iu \xi)
- {\mathfrak K}(s, v', t,v, \iu \xi) \right] \\
&= \left(\frac{v}{v_l}\right)^{\varepsilon v_l}  e^{-\varepsilon \left[v - v_l + (t - t_k)^2 \right]} \int_0^\infty  \xi \left[e^{0} - e^{-0}\right] d\xi = 0. \nonumber
\end{align}
Also, when $s - t \ll 1$, from \eqref{IkiNew} we have
\begin{align}
a_2 &\approx \frac{g^2(s,p)}{2 (\tau(t,p) - \tau(s,p))}, \qquad a_5 \approx \frac{g^2(t,p) v}{2 (\tau(t,p) - \tau(s,p))},
\end{align}
\noindent and the integrand in \eqref{IkiNew} becomes proportional to
\begin{align*}
 \left(\frac{v}{v_l}\right)^{\varepsilon v_l} e^{-\varepsilon (s - t_k)^2 - \varepsilon (v-v_l)} \int_0^\infty d\xi \, \xi  e^{[y(y) - y(s)]\sqrt{p} + \beta'(t,p)(t-s) + \alpha(t) q(t) v (t-s)} \left(\frac{g(t,p)}{g(s,p)}\right)^2 (\tau(t,p) - \tau(s,p)).
\end{align*}
Here $q(t) = 2\log'(g(t)) - \log'(\alpha(t)) > 0$, and
\begin{equation}
\alpha(t) q(t) = - \frac{1}{2} \alpha^2(t) \sigma^2(t) + c = - \frac{1}{2} \alpha^2(t) \sigma^2(t)  + \frac{1}{2}(-\xi^2 \pm \iu \xi),
\end{equation}
\noindent so $\Ree(\alpha(t) q(t)) < 0$. Thus, the integral $I_{kl}(t,v)$ is well-behaved in this limit.

Also $\Ree(a_2) < 0$. Indeed,  by definition in \eqref{aCoeff}
\begin{equation}
a_2 = \varepsilon + \frac{g^2(s,p)}{2 (\tau(t,p) - \tau(s,p))} + \alpha(s,p),
\end{equation}
\noindent and as shown in Appendix~\ref{app1}, $\Ree(\alpha) \le 0$, while the definition of $\tau(t,p)$ in \eqref{tauTr} implies $\tau(t,p) <  \tau(s,p)$ for $t < s$.

\subsubsection{Computation of oscillating integrals}

Combining all the results obtained in the previous section, the integrals in \eqref{IkiNew} can be finally represented in the form
\begin{align} \label{IkiNew2}
I_{kl}(t,v) &=  A \int_t^T d s \,  e^{-\varepsilon (s - t_k)^2} \int_0^\infty d \xi\, \frac{\xi}{\zeta(t,s,p)}
\left(\frac{g(s,p)}{g(t,p)} \right)^{b-1/2} \left(\frac{\varpi_3}{\varpi_2}\right)^{b+3/2} \left(b + \frac{1}{2} + \frac{\varpi_3^2}{4 \varpi_2} v\right) \\
&\times e^{- y(s)\sqrt{p} + \beta(t,p) - \beta(s, p) + \varsigma(t,s,p) v}, \nonumber \\
A &=  2^{-1/2 - b} \left(\frac{v}{v_l}\right)^{\varepsilon v_l} e^{-\varepsilon (v-v_l)},
\quad \varpi_2 = \varepsilon + \alpha(s,p) + \frac{g^2(s,p)}{g^2(t,p)} \zeta(t,s,p), \quad
\varpi_3 =  2\frac{g(s,p)}{g(t,p)} \zeta(t,s,p), \nonumber \\
\zeta(t,s,p) &= \frac{g^2(t,p)}{2 (\tau(t,p) - \tau(s,p))}, \quad \varsigma(t,s,p) = \alpha(t,p) -  \zeta(t,s,p) + \frac{\varpi_3^2}{4 \varpi_2}.  \nonumber
\end{align}
The second integral in this expressions is oscillating since it contains complex exponents. Same is true for the integral in the RHS of the definition of $P_1$ in \eqref{finSol}.

There exists a vast literature on computing numerically integrals of the type $\int_{a}^{b} f(x) \mathrm{e}^{\mathrm{i} \omega x} \mathrm{~d} x$. When $\omega$ is large, the integrand becomes highly oscillatory and conventional quadrature programs are ineffective. Various methods have been proposed to address this, mainly based on Filon's algorithm and its modifications, see \citep{Floch2018} and references therein.  For instance, in \citep{Shampine2011} a new method based on a smooth cubic spline is implemented in MATLAB that is both easy to use and effective for large $\omega$. Because the implementation of the basic method is adaptive, the program deals comparatively well with $f(x)$ that have peaks. With the assistance of another method, the program is able to deal effectively with $f(x)$ that have a moderate singularity at one or both ends of $[a, b]$.

Then in \citep{Shampine2012} more complicated integrals $\int_{a}^{b} f(x) e^{i g(x)} d x$ were considered with $g(x)$ being large on $[a, b]$. Previous approaches require users to supply the location and nature of critical points of $g(x)$ and may require $g^{\prime}(x)$. With the new approach proposed in the paper, the program \verb|quadgF| merely asks a user to define the problem, i.e., to supply $f(x), g(x), [a, b]$, and specify the desired accuracy. Though intended only for modest relative accuracy, \verb|quadgF| is very easy to use and solves effectively a large class of problems.

However, our integrals have a more complicated form, hence the above approach cannot be used. Therefore, we rely on a different idea which is exploited in  \verb|Chebfun| package, \citep{Trefethen}. \verb|Chebfun| is an open-source package for computing with functions to about 15-digit accuracy.  The implementation of \verb|Chebfun| is based on the mathematical fact that smooth functions can be represented very efficiently by polynomial interpolation in Chebyshev points, or equivalently, thanks to the Fast Fourier Transform, by expansions in Chebyshev polynomials.
Accordingly, the integrals are normally calculated by an FFT-based version of Clenshaw-Curtis quadrature, as described first in \citep{Gentleman1972}. This formula is applied on each piece of the function (i.e., each smooth piece of the \verb|Chebfun|), and then the results are added up. Various examples provided in \citep{Trefethen} demonstrate the efficiency of this approach in computing oscillating integrals, therefore, we use it in this paper.

Alternative, our experiments show that utilization of Gauss-Kronrod quadratures (\verb|quadgk| function in MATLAB) in our case provides the results that are very close to those of  \verb|Chebfun|) but the elapsed time is lower. Therefore, in all tests we finally make use of \verb|quadgk|. Despite this function is capable to work with the infinite upper limit of the integral, in our case the accuracy of results is not sufficient. Perhaps, various tweaks and tuning can solve this problem, however, here we use truncation of the infinite interval to the corresponding upper limit $\Upsilon = \max(\xi)$. Our experiments indicate, that for relatively short maturities $T < 1$ choice of $\Upsilon = 500$ is good, while for long maturities $T \ge 1$ we have to increase it to  $\Upsilon = 5000$, and for $T \ge 2$  - even more up to $\Upsilon = 20000$.

\section{Numerical experiments} \label{numExp}

In this section we present results of our numerical tests aimed to check the accuracy and speed of the proposed approach. For doing so an explicit form of the model parameters $\kappa(t), \theta(t), \sigma(t), \rho(t), v_0$ should be specified. Since we don't calibrate the model to market quotes, in these tests (without any loss of generality) we choose an artificial (test) dependencies, namely:
\begin{equation} \label{ex}
\theta(t) = \theta_0 e^{-\theta_k t}, \qquad \sigma(t) = \sigma_0 e^{-\sigma_k t},
\qquad \rho(t) = \rho_0, \qquad \kappa(t) = \frac{m \sigma^2(t)}{2 \theta(t)},
\end{equation}
\noindent with $m, \theta_0, \sigma_0, \rho_0, \theta_k, \sigma_k$ being constants.

We compare our results with those obtained by solving \eqref{PDE1} using the FD method described in detail in \citep{Itkin2014b}. In short, this ADI (alternative direction implicit) scheme is of the second order in all dimensions, uses few first Rannacher steps on a non-uniform grid compressed close to the spot $S_0$ and the initial instantaneous variance $v_0$. Parameters of the test are presented in Table~\ref{param}, and a typical FD grid - in Fig.~\ref{grid}.
\begin{table}[!htb]
\begin{center}
\begin{tabular}{|c|c|c|c|c|c|c|c|c|c|c|c|c|}
\hline
$S_0$ & $m$ & $\theta_0$ & $\sigma_0$ & $\rho_0$ & $\theta_k$ & $\sigma_k$ & $L$   & $v_0$ & $r$ & $q$ \\
\hline
60 & 2 & 0.1 & 0.3 & -0.7 & 0.3 & 0.2 & 40 & 0.5 & 0.02 & 0.01  \\
\hline
\end{tabular}
\caption{Parameters of the test.}
\label{param}
\end{center}
\end{table}

\begin{figure}[!htb]
\begin{center}
\hspace*{-0.3in}
\fbox{\includegraphics[width=0.6\textwidth]{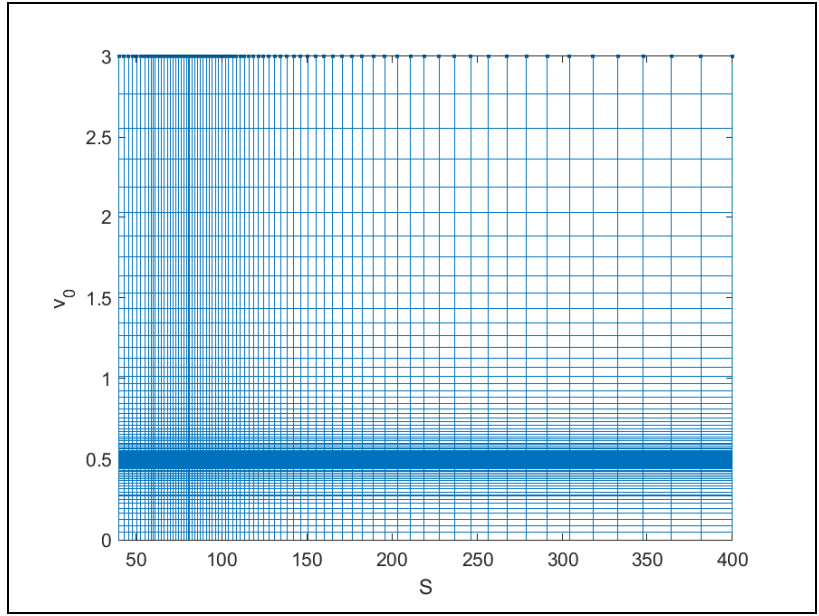}}
\end{center}
\caption{A typical nonuniform grid with 76 nodes in $S$ and 79 nodes in $v$ for the FD scheme with $S_0, \ v_0$ given in Table~\ref{param}.}
\label{grid}
\end{figure}

Validation of the FD method can be done for a European vanilla Put in the Heston model with constant parameters $\kappa_0, \theta_0, \sigma_0, \rho_0, v_0$ since for this model the Put price can be found by FFT. The FFT price computed by using 8192 nodes is 24.9381 and the FD price on the above FD grid is 24.9378 or 12 bps of difference.

To solve the Volterra equation in \eqref{volt} by the RBF method, as the collocation points, similar to \citep{ItkinMuraveySABR},  we choose a uniform grid in $t \in [0,T]$ and $v \in [v_0 - v_m, v_0 + v_m]$ with $V_m = 0.1$. We take $N_t = 10, N_v = 4$. The integrals in time are computed by using the Simpson quadratures. Also, following the discoveries in \citep{ItkinMuraveySABR}, to solve a system of linear equations obtained via the RBF method we  use a \verb+minres+ iterative solver which is good when the matrix is not positive definite, but symmetric. Since here we constructed a non-standard (non-Gaussian) RBF, strictly speaking our matrix is not symmetric. However, our experiments show that it is almost symmetric with the maximum absolute difference between the corresponding elements $|a(i,j) - a(j,i)| \approx 0.001$. Also, our experiments with various iterative solvers clearly indicate that \verb+minres+ provides the best results. The main advantage of this method lies in the fact that it makes it possible to construct an orthogonal basis for the Krylov subspace by three-term recurrence relations, \citep{minres1975}.  We do see that due to small rounding errors some eigenvalues of our RBF matrices are either very small negative numbers or zero. Therefore, in principle, the RBF matrix should be first regularized by using a standard procedure.  However, again, since this error depends on quality of the RBF interpolation it is expected that modern methods which are more stable than the global RBF method would provide better accuracy. Indeed, for the Gaussian RBF method a typical matrix in \eqref{linSys}  has the condition number about $10^{18}$ and, therefore, even iterative methods of solving this system of linear equations produce bigger error than in the case of a well-conditioned matrix.

The best values of $\varepsilon$ in all experiments are given in Table~\ref{tabEps}. It can be seen that these values are almost independent of the strikes and maturities. Surprisingly, these values of $\varepsilon$ are not small as compared with those used in traditional RBF methods. This can be caused by two reasons: i) we use a new non-standard RBF method, and for this method typical values of $\varepsilon$ are not well investigated; ii) we use  approximation of the Kummer function in \eqref{KumApp}, so the results could be sensitive to the choice of $\varepsilon$. In out experiments $b = 1.5, v_0 = 0.5$, so the assumption $b+3/2 \gg \varepsilon v_0$ doesn't entirely hold (i.e., $\varepsilon v_0/(b+3/2) = 0.66$ is not a very small parameter). Therefore, we verified this approximation by using the Kummer function itself instead of the approximation in \eqref{KumApp}, and didn't observe any significant difference in the final results. Thus, at the moment we attribute these values of $\varepsilon$ as being  inherent to the new method.

\begin{table}[!htb]
  \centering
      \begin{tabular}{|l|r|r|r|r|r|r|}
    \toprule
     $K$ & 45 & 50     & 60     & 70     & 80     & 90 \\
\hline
     $\varepsilon$ & 3. & 5.  & 4.  & 4.  & 4.  & 4. \\
    \bottomrule
    \end{tabular}%
\caption{The values of $\varepsilon$ used in the numerical experiments.}
  \label{tabEps}%
\end{table}%

We run the test for a set of maturities $T \in [1/24, 1/12, 0.25,0.5,1,2]$ years and strikes $K \in [45, 50, 60, 70, 80, 90]$. The Down-and-Out barrier Put option prices computed in these experiments are presented in Tab.~\ref{TabComp}. For the FD method the time step is fixed and equal to 0.01 year to preserve the method accuracy in time. Typical elapsed times are also shown in Tab.~\ref{TabComp}. We emphasize that the GIT time shows the total time for all strikes, while the FD time shows the time for one strike (since we run a backward scheme). For the forward scheme this time should be increased since after the density is found by the FD method additional integration with the payoff function for all strikes should be done.  Also in Tab.~\ref{Error} the relative percentage error between the FD and GIT solutions is presented.

\begin{table}[!htb]
  \centering
  \scalebox{0.77}{
    \begin{tabular}{|c!{\vrule width 1.5pt}|r|r|r|r|r|r!{\vrule width 1.5pt}r|r|r|r|r|r|}
    \specialrule{.1em}{.05em}{.05em}
   \hfill $\bf{T}$     & 0.042 & 0.083 & 0.25  & 0.5   & 1 & 2  & 0.042 & 0.083 & 0.25  & 0.5   & 1 & 2\\
    \specialrule{.1em}{.05em}{.05em}
    $\bf{K}$     & \multicolumn{6}{c!{\vrule width 1.5pt}}{\textbf{GIT}} & \multicolumn{6}{c|}{\textbf{FD}} \\
    \specialrule{.1em}{.05em}{.05em}
    45    & 0.0343 & 0.0466 & 0.0352 & 0.0227 & 0.1298 & 0.5616 & 0.0288 & 0.0621 & 0.0450 & 0.0252 & 0.0763 & 0.7668 \\ \hline
    50    & 0.2760 & 0.4390 & 0.3402 & 0.0932 & 0.1642 & 1.2165 & 0.3187 & 0.5249 & 0.3602 & 0.1967 & 0.1596 & 0.7836 \\ \hline
    60    & 2.9707 & 3.3150 & 2.4670 & 1.3160 & 0.3904 & 1.0277 & 3.3253 & 3.6601 & 2.2717 & 1.2213 & 0.6829 & 1.0200 \\ \hline
    70    & 10.0135 & 9.5921 & 6.9704 & 4.0235 & 1.6431 & 1.1620 & 10.3518 & 9.7413 & 5.8385 & 3.3366 & 1.7720 & 1.1771 \\ \hline
    80    & 19.3932 & 18.1225 & 13.3700 & 8.9262 & 4.3228 & 1.6871 & 19.6622 & 17.8838 & 10.8061 & 6.3585 & 3.4695 & 1.9798 \\ \hline
    90    & 29.2469 & 27.4915 & 20.9084 & 14.8801 & 7.5348 & 2.7396 & 29.4417 & 26.8270 & 16.6750 & 10.1083 & 5.6790 & 3.2311 \\
  \specialrule{.1em}{.05em}{.05em}
    Elapsed time & 2.45      & 2.02    & 1.91   &  1.87   & 2.19  & 2.53  & 0.13 & 0.23  & 0.65  & 1.3   & 2.7 & 5.4  \\
    \specialrule{.1em}{.05em}{.05em}
    \end{tabular}%
  }
  \caption{Comparison of Down-and-Out Put option prices for the Heston model obtained by the FD and GIT methods and elapsed time in secs.}
  \label{TabComp}%
\end{table}%

\begin{table}[!htb]
\centering
\begin{tabular}{|c|r|r|r|r|r|r|}
\specialrule{.1em}{.05em}{.05em}
& \multicolumn{6}{c|}{\textbf{T}} \\
\specialrule{.1em}{.05em}{.05em}
\textbf{K} & 0.042 & 0.083 & 0.25 & 0.5 & 1.0 & 2.0 \\
\specialrule{.1em}{.05em}{.05em}
 45 & -19.10 & 24.96 & 21.78 & 9.92  & -70.12 & 26.76 \\ \hline
50  & 13.40 & 16.37 & 5.55  & 52.62 & -2.88 & -55.25 \\ \hline
60  & 10.66 & 9.43  & -8.60 & -7.75 & 42.83 & -0.75 \\ \hline
70  & 3.27  & 1.53  & -19.39 & -20.59 & 7.27  & 1.28 \\ \hline
80  & 1.37  & -1.33 & -23.73 & -40.38 & -24.59 & 14.78 \\ \hline
90  & 0.66  & -2.48 & -25.39 & -47.21 & -32.68 & 15.21 \\
\specialrule{.1em}{.05em}{.05em}
\end{tabular}%
\caption{The relative percentage error of the GIT solution as compared with the FD one.}
\label{Error}%
\end{table}%

Looking into Tab.~\ref{Error} one can see that for our problem the relative error of the GIT method as compared with the FD reference solution varies across strikes and maturities. For large strikes and short maturities the error is of order of few percents, while for intermediate maturities it is in the range [20\%,40\%]. However, for some strikes, say ATM, it is about 8\%.  For large maturities, as mentioned, we need to use more integration points while the error varies from few percents and up to 2-30\% depending on the strike. Big relative error at $K=45$ and high maturities is due to the small price value, hence even small absolute errors could produce high relative errors.

\section{Discussion} \label{Discus}

In this paper we proposed a generalization of the GIT method to price Down-and-Out barrier Put options $P_\mathrm{do}$ under the Heston stochastic volatility model where all coefficients and the barrier are deterministic functions of the time (subject to the condition \eqref{cond1}). The method requires solving a two-dimensional mixed Volterra-Fredholm equation for the gradient $\Phi\left(t, v \right)$ of the solution at the moving boundary $x = y(t)$. Once it is found, the option price $P_\mathrm{do}$ follows since it was expressed in a semi-analytical form via a two-dimensional integral of $\Phi\left(t, v \right)$. Note, that this integral is computed as a part of the system matrix $\|A\|$ for the LMVF equation, and hence doesn't require extra time.

Note that barrier options trade in over-the-counter markets for many reasons. In this paper we focus on a Down-and-Out Put, but in-out parity implies that our results easily apply to a Down-and-In Put as well. Indeed, by using the barrier options parity, the price of the Down-and-In barrier Put option $P_\mathrm{di}$ can be found as $P_\mathrm{di} = P_\mathrm{van} - P_\mathrm{do}$, where $P_\mathrm{van}$ is the price of the European vanilla Put option in the Heston model. Since for the Heston model a closed-form solution for European options (via an FFT transform) is known, our solution for $P_\mathrm{di}$ also provides a closed form solution for $P_\mathrm{di}$. For the Up-and-Out barrier Put option $P_\mathrm{uo}$ a simple change of variables $x \to -x$ reduces the pricing problem to that one which we consider in this paper. Therefore, the price of the Up-and-In barrier Put option can be found by using the barriers option parity. The Call options can be priced in a similar way by using a covered Call instead of a Put.

One of the possible reasons that an investor might buy a Down-and-In Put is to lock in the premium. Suppose that an investor plans to buy a vanilla Put for some strike $K$ and maturity $T$ if the underlying drops to some level $L < K$  before $T$. Even if we condition on the underlying asset's price being below $L$ before $T$, the premium that will be paid for the vanilla Put is random due to stochastic volatility and due to the possibility that the barrier is crossed, not touched. In contrast, if an investor buys a Down-and-In Put initially instead, then a known premium is paid at inception and there are no further cash payments when the barrier is touched or crossed. The investor is in essence pre-paying to remove a random entry cost.

As shown at the end of Section~\ref{solVolt}, since the dependence of $\Phi\left(t, v \right)$ on the strike $K$ appears only in the LHS of the LMVF equation \eqref{volt}, the gradient $\Phi\left(t, v \right)$ for different strikes can be found by solving a single system of linear equations with multiple RHS. Also, taking $T$ large enough (for stock and index options traded at the market $T \le 1$ year, so we can choose, e.g. $T < T_* =2$ years) one can solve the LMVF equation, and find $\Phi\left(t, v \right)$ for all $t \le T_*$ in one sweep. Then the barrier option prices can be obtained for all maturities $T \le T_*$ by computing the RHS in \eqref{finSol}. Therefore, our method is similar to solving the forward PDE (where the density of the underlying can be found in one sweep and then the option prices for various $K$ and $T$ come by integrating this density with the payoff), rather than the backward one.

 As far as the method performance is concerned, for the GIT method the elapsed time in Tab.~\ref{TabComp} represents the computational time for one strike. As we mentioned, the method could be naturally parallelized since the RBF matrix is independent of the strike values. Therefore, in principle, this time should be divided by the number of strikes (seven in our case). But for the FD method the elapsed time also represents the computational time for the backward scheme. Switching to the forward scheme (solution of the forward equation) could reduce this time in the same way as the GIT elapsed time. Therefore, regardless whether we use the forward scheme of the backward scheme the elapsed time of both methods is of the same order at high maturities, while at small maturities the FD method is faster.

However, performance of our method is totally determined by the performance of the numerical method we use to solve the LMVF equation. We see that computing elements of the matrix $\|A\|$ takes 90\% of the total elapsed time as it should be since this matrix is dense (a known deficiency of the global RBF method). Obviously, by using localized versions of the RBF method, the elapsed time can be significantly improved. However, we don't consider this approach in detail in this paper and leave it for future research. Also, better methods of computing oscillating integrals with less number of points could significantly accelerate our approach. For instance, when changing $\Upsilon$ from 500 to $50$, a typical elapsed time drops down to 0.2 secs, i.e. becomes comparable with that one of the FD method at short maturities (for long maturities we still need more points).

In our numerical experiments we used the number of points in time $M=10$. We emphasize that the value of M is actually determined by steepness of the time-dependent coefficients. If the coefficients are smooth in time, the choice of $M=10$ is sufficient. However, if they are a fast-varying functions of the time, one has to take more temporal points. The latter will definitely slow down the method performance. Again, this depends on what kind of the numerical method for solving the LMVF equation is in use, e.g. the global vs a localized RBF.

A formal (theoretical) comparison of our approach with the FD method reveals the following. The FD method requires a 3D grid for temporal $t$ and two spatial $x,v$ variables. In our method, since we derived a semi-analytical expression for the barrier  option price, we need a 2D grid in $(t,v)$ to solve the LMVF equation numerically. Therefore, we dropped off one dimension that gives rise to acceleration of computations. On the other hand, we have to compute highly oscillating integrals that may take time. Also, integrands in the LMVF equation require computation of elementary functions, like $\sin, \cos, \tan, \exp$ while computing a FD matrix requires just simple operations. In both methods the system matrix can be banded: for the FD method this is very natural; for the RBF method this can be achieved by using a localized version of the method. Also, in our method we do integration in time by using high-order quadratures (the Simpson rule) with accuracy $O((\Delta t)^4)$ while the FD method usually provides $O((\Delta t)^2)$. Therefore, we can reduce the number of points in time as compared with the FD grid. Certainly, meshless (e.g., RBF) numerical methods could also be used for solving the pricing PDE. Then the main difference of two approaches remains the same: our problem has one dimension less, but requires computing oscillating integrals dependent on some elementary functions.

As we have already mentioned in various papers about the GIT method (see, e.g., \citep{ItkinLiptonMuraveyBook}),
computation of option Greeks can be done in a similar manner as the option prices.  That is because the GIT method provides the option price in a semi-analytical form (via integrals). Therefore, the explicit dependence of prices on the model parameters is available via differentiation of the option price with respect to a necessary parameter (a simple differentiation under the integrals). Thus, the values of Greeks can be calculated simultaneously with the prices almost with no increase in time. Indeed, differentiation just slightly changes the integrands, and these changes could be represented as changes in weights of the quadrature scheme used to numerically compute the integrals. However, from the computational speed point of view the most challenged piece is computation of densities which contain special functions. These densities can be saved and then reused for computation of Greeks.

Finally, the proposed method can also be applied to any uncorrelated SV model if the Green's function of the instantaneous variance process is known in closed form.  Here we employed \eqref{Green} - the Green's function of the one-dimensional Bessel process since the CIR model for $v_t$ can be transformed to this process. But other popular choices, e.g. the lognormal process can be treated in the same way. Thus, our approach is general enough to deliver semi-analytical prices of barrier options for many SV models.

However, for the correlated SV model we strongly depend on the exponential form of the GIT in \eqref{GITdef}. If a similar (exponential) transform can be constructed to obtain a closed form representation of the image $\bar{u}(t,v,p)$ (e.g., that one in \eqref{sol1}), then our machinery should work given the Green's function of the instantaneous variance is known. Otherwise, this remains to be an important yet open question whether this is possible.

\section*{Acknowledgments}

We are grateful to Alex Lipton and Fazlollah Soleymani for some useful discussions. Dmitry Muravey acknowledges support by the Russian Science Foundation under the Grant number 20-68-47030.



\appendixpage
\appendix
\numberwithin{equation}{section}
\setcounter{equation}{0}

\section{Derivation of the \eqref{sol1}} \label{app2}

Based on the solution for $\barU(\tau, z)$ found in \eqref{sol}, and the definition of $\barU(\tau, z)$ in \eqref{trCIR1}, the second integral in \eqref{sol} can be represented as

\begin{align*}
I_2(t,z, p) &= e^{\alpha(t, p) v + \beta(t,p)}\int_0^\tau \int_0^\infty G(\tau-k, z, \zeta) \Psi(k,\zeta,p) dk \, d\zeta.
\end{align*}
To return to the original variables $(t,v)$ we make transformations
\begin{align*}
k \mapsto \frac{1}{4}\int_{s}^T g^2(\gamma,p) \sigma^2(\gamma) d\gamma,
\quad dk \mapsto - \frac{1}{4} g^2(s,p) \sigma^2(s) ds,
\end{align*}
\noindent and recall that
\begin{equation*}
\Psi(t, z, p) = \frac{\Phi_1(t,z,p)}{\sigma^2(t)} =
-\frac{4 v}{g^2(t,p) \sigma^2(t)} e^{-y(t) \sqrt{p} - \beta(t,p) - \alpha(t,p) v} \Phi(t,v), \quad v = \frac{z^2}{g^2(t,p)}, \ dz = \frac{1}{2\sqrt{v}} g(t,p) dv.
\end{equation*}
Therefore,
\begin{align*}
I_2(t,v,p) &= - \frac{1}{4} e^{\alpha(t,p) v + \beta(t,p)} \int_t^T \int_0^\infty G\left(\tau - \frac{1}{4}\int_s^T g^2(\gamma, p)\sigma^2(\gamma) d\gamma, z, z'\right) g^2(s,p) \Phi_1(s,z',p) ds \,dz' 	\nonumber 	\\
&= e^{\alpha(t,p) v + \beta(t,p)}\int_t^T \int_0^\infty v' \Bigg\{ G\left(\frac{1}{4}\int_t^s g^2(\gamma, p)\sigma^2(\gamma) d\gamma, \sqrt{v}g(t,p), z'\right)	\\
&\times e^{-y(s)\sqrt{p}- \left[\beta(s, p) + \frac{(z')^2}{g^2(s,p)}\alpha(s,p)\right]} \Phi\left(s, \frac{(z')^2}{g^2(s, p)}\right) \Bigg\} ds \, dz' \nonumber \\
&= \frac{1}{2}\int_t^T \int_0^\infty \sqrt{v'} g(s,p) \Bigg\{ G\left(\int_t^s \frac{1}{4} g^2(\gamma, p)\sigma^2(\gamma) d\gamma, g(t,p) \sqrt{v}, g(s, p) \sqrt{v'} \right) 	\\
&\times e^{-y(s)\sqrt{p} +\alpha(t,p) v + \beta(t,p)- \left[\beta(s, p) + v'\alpha(s,p)\right]} 	\Phi\left(s, v' \right) \Bigg\}ds \, dv'. \nonumber
\end{align*}

\section{Connection to the Fourier-sine transform} \label{appSinFT}

The classical Fourier-sine transform can be applied to functions defined on the positive real semi-axis. However, it can be easily generalized to the functions defined on $[y(t),\infty)$. Indeed, using the simple phase shift $x \mapsto x + y$ yields the following transform
\begin{equation} \label{shiftedSinedef}
\bar v(\xi) = \frac{2}{\pi} \int_{y(t)}^\infty v(x) \sin \left(\xi [x-y(t)]\right)  dx, \qquad v(x) =  \int_0^\infty \bar v(\xi) \sin \left(\xi [x-y(t)]\right)  dx.
\end{equation}
The integral in \eqref{invTr} is the inverse transform of the type \eqref{shiftedSinedef}, therefore the function $\chi(\xi ,t, v)$ can be found by the direct transform
\begin{equation*} \label{chiIFTdef}
	\chi(\xi ,t, v) = \frac{2}{\pi} \int_{y(t)}^\infty  P(t,x,v) \sin \left(\xi [x-y(t)]\right ) dx.
\end{equation*}
Applying the connection formula between the sine and the hyperbolic sine
\begin{equation*}
\chi(\xi,t, v) = \frac{1}{ \iu \pi }\int_{y(t)}^\infty P(t,x,v) \left[ e^{\iu \xi [x-y(t)]} - e^{-\iu \xi [x-y(t)]} \right] dx,
\end{equation*}
\noindent and using \eqref{GITdef} with $\sqrt{p} = \pm \iu \xi$ we obtain \eqref{chi1}.

\section{Behavior of the solution \eqref{finSol} at $\bm{v \to \infty}$.} \label{app1}

Obviously, the solution $u(t,x,v)$  in \eqref{finSol} should be finite at $v \to \infty$. This can be achieved by choosing an appropriate terminal condition $\alpha(T,p)$. Below we deal with both integrals in the RHS of \eqref{finSol} and analyze them separately.

\paragraph{The first integral.} Let us show that the terminal condition $\alpha(T,p) = 0$ is sufficient for the first integral to converge. With this condition we have $B_1=0$ and, hence, the only term that depends on $v$ is $e^{\gamma(t,p) v}$ where $ \sqrt{p} = -\iu \xi$.

Suppose that $\alpha(t,p)$, which solves the Riccati equation \eqref{ric2}, can be represented as $\alpha(t,p) = \alpha_R(t,p) + \iu \alpha_I(t,p), \ \alpha_R \in \mathbb{R}, \ \alpha_I \in \mathbb{R}$. Also, suppose that the remaining part of the integrand under the first integral (without the term $e^{\gamma(t,p) v}$) can be represented as $C(t,p) = C_R(t,p) + \iu C_I(t,p), \ C_R \in \mathbb{R}, \ C_I \in \mathbb{R}$. By simple arithmetic the imaginary part of the whole integrand is $e^{\alpha_R v} [C_I \cos( \alpha_I v) + C_R \sin(\alpha_I v)]$. Thus, the first integral converges if $\alpha_R \le 0$.

Let consider \eqref{ric2} when $\sqrt{p} = -\iu \xi$. Since $\alpha(t,p)$ is complex, the Riccati equation can be written separately for the real $\alpha_R$ and imaginary $\alpha_I$ parts. It is easy to see that $\alpha_R(t,e^{-\iu \pi} \xi^2)$ solves the equation
\begin{align} \label{ricAlphaR}
\alpha'_R &= \frac{1}{2} D(t,\xi) + \kappa(t) \alpha_R - \frac{1}{2}\sigma^2(t) \alpha^2_R, \quad \alpha_R(T) = 0, \\
D(t,\xi) &= \xi^2 + 2 \rho(t) \sigma(t) \xi \alpha_I +  \sigma^2(t) \alpha_I^2 \in \mathbb{R}. \nonumber
\end{align}

Observe that $D(t,\xi) \ge 0$ because $\rho(t) \in [-1,1]$ and
\begin{equation*}
(\xi - \sigma(t) \alpha_I)^2 \le D(t,\xi) \le (\xi + \sigma(t) \alpha_I)^2.
\end{equation*}
Also, observe that if $D(t,\xi)=0$, \eqref{ricAlphaR} has a closed form solution $\alpha_R(t,\xi) = 0$. This is because
this solution obeys both the equation and the terminal condition (and is the reason why this terminal condition has been chosen). Now, it can be checked that returning the term $D(t,\xi) \ge 0$ back into \eqref{ricAlphaR} we decrease the solution for $t < T$. Since by the terminal condition $\alpha(T,\xi) = 0$, this means that $\alpha_R(t,\xi)$ is nonpositive $\forall t \in [0,T]$. Thus, the first integral well behaves at $v \to \infty$. Exactly same analysis is valid for
$\sqrt{p} = \iu \xi$.

\paragraph{The second integral.} In the second integral we have two competitive terms which depend on $v$, this is $e^{\alpha(t,p) v}, \ \sqrt{p} = \iu \xi$ and the Green function as a function of $g(t,p) \sqrt{v}$. Using the definition of the Green function in \eqref{Green} and computing its asymptotic at $v \to \infty$ we obtain, \citep{as64}
\begin{equation*}
\lim_{v \to \infty} G\left(\frac{1}{4}\int_t^s g^2(\zeta, p) \sigma^2(\zeta) d \zeta, \, g(t,p) \sqrt{v},\, g(s, p)
\sqrt{v'} \right) \propto \exp\left[- \frac{1}{2 \tau} g^2(t, e^{-\iu \pi} \xi^2) v \right].
\end{equation*}
Again, using the same logic as for the first integral and having in mind that $\alpha_R \le 0, \ \Ree (g^2(t,p)/\tau) > 0$ we can conclude that this integral also converges at $v \to \infty$.

\section{Positive definiteness of the function $\bar{\Theta}(t,\nu)$} \label{appPD}

Here we prove that the basis function $\bar{\Theta}(t,\nu)$ proposed in \eqref{newGauss} is positive definite. As per \citep{Fasshauer}, a complex-valued continuous function $\Theta: \mathbb{R}^{s} \rightarrow \mathbb{C}$ is called positive definite on $\mathbb{R}^{s}$ if
\begin{equation} \label{defPD}
\sum_{j=1}^{N} \sum_{k=1}^{N} c_{j} \overline{c_{k}} \Theta\left(\boldsymbol{x}_{j}-\boldsymbol{x}_{k}\right) \geq 0,
\end{equation}
\noindent for any $N$ pairwise different points $\boldsymbol{x}_{1}, \ldots, \boldsymbol{x}_{N} \in \mathbb{R}^{s}$, and $\boldsymbol{c}=\left[c_{1}, \ldots, c_{N}\right]^{T} \in \mathbb{C}^{N}$. The function $\Theta$ is called strictly positive definite on $\mathbb{R}^{s}$ if the quadratic form \eqref{defPD} is zero only for $\boldsymbol{c} \equiv \mathbf{0}$.

By the Bochner theorem \citep{Bochner1932}, a (complex-valued) function $\Theta \in C\left(\mathbb{R}^{s}\right)$ is positive definite on $\mathbb{R}^{s}$ if and only if it is the Fourier transform of a finite non-negative Borel measure $\mu$ on $\mathbb{R}^{s}$. Real valued functions are a special case of this theorem which is covered by the following
Corollary, \citep{Fasshauer, Wendland2005}:
\begin{corollary}
Let $f$ be a continuous non-negative function in $L_{1}\left(\mathbb{R}^{s}\right)$ which is not identically zero. Then the Fourier transform of $f$ is strictly positive definite on $\mathbb{R}^{s}$.
\end{corollary}

The proof is based on the fact that this is a special case of the Bochner theorem in which the measure $\mu$ has Lebesgue density $f$. Thus, we use the measure $\mu$ defined for any Borel set $B$ by
\[ \mu(B)=\int_{B} f(\boldsymbol{x}) \mathrm{d} \boldsymbol{x}. \]
Then the carrier of $\mu$ is equal to the (closed) support of $f$. However, since $f$ is non-negative and not identically equal to zero, its support has positive Lebesgue measure, and hence the Fourier transform of $f$ is strictly positive definite by the Bochner theorem.

Using the above facts, we formulate the following statement
\begin{theorem}
Suppose variables $[t, T, \nu, \nu_l, \epsilon]$ are real, $[\nu, \nu_l, \epsilon] \in [0,\infty), \ t \in [0,T], \ T > 0]$  and consider a function $\bar{\Theta}_{kl}(t,\nu): \mathbf{R} \to \mathbf{R}$ defined as
\begin{equation} \label{newGauss1}
\bar{\Theta}_{kl}(t,\nu) = \left(\frac{\nu}{\nu_l}\right)^{2 \varepsilon \nu_l^2}  e^{-\varepsilon \left[\nu^2 - \nu_l^2 + (t - t_k)^2 \right]} + w \delta(\nu),
\end{equation}
\noindent where $\delta(\nu)$ is the Dirac delta function, and $w > 0$. Then the coefficient $w$ can be chosen such that $\bar{\Theta}_{kl}(t,\nu)$ is positive definite.
\end{theorem}

\begin{proof}
Let $\omega \in \mathbf{R}$ and consider the function
\begin{align}
F(\omega) &= \frac{1}{\sqrt{2 \pi}} + A \varepsilon^{-\frac{1}{2} - \varepsilon \nu_l^2} \Gamma \left(\frac{1}{2}  + \varepsilon  \nu_l^2 \right) \, M\left(\frac{1}{2}  + \varepsilon  \nu_k^2; \frac{1}{2}; - \frac{\omega^2}{4\varepsilon} \right), \\
A &= \frac{1}{\sqrt{2 \pi}} e^{-\varepsilon [(t-t_k)^2 - \nu_l^2]} {\nu_l}^{-2 \varepsilon \nu_l^2}, \nonumber
\end{align}
\noindent where $\Gamma(x)$ is the gamma function, and $M(a,b,z)$ is the Kummer confluent hypergeometric function, \citep{as64}.
It can be directly checked that the Fourier transform ${\cal F}(F: \omega \to \nu)$ of $F(\omega)$ is equal to $\bar{\Theta}_{kl}(t,\nu)$. Therefore, to prove the theorem we need to show that $F(\omega)$ takes all values in $[0,\infty)$.

Since $\varepsilon > 0, \nu_l \ge 0$ we have $A > 0$ and $\Gamma \left(\frac{1}{2}  + \varepsilon  \nu_l^2 \right) > 0$, the minimum of $F(\omega)$ in $\omega$ is reached at the first zero of the Kummer function $M\left(\frac{3}{2}  + \varepsilon  \nu_l^2; \frac{3}{2}; - \frac{\omega^2}{4\varepsilon} \right)$ which is the derivative of $M\left(\frac{1}{2}  + \varepsilon  \nu_l^2; \frac{1}{2}; - \frac{\omega^2}{4\varepsilon} \right)$ with respect to $\omega$ (note that $F(\omega)$ is even). This is because when $\omega$ increases,  $F(\omega)$ rapidly vanishes
\footnote{
Note, that, e.g., in Wolfram Mathematica the value of $M(a,b,x)$ is computed incorrectly, when  $d$ is a high negative integer and $z$ is high. We remind that by using Kummer transformation, we have
\begin{equation*}
M\left(\frac{1}{2}  + \varepsilon  \nu_l^2; \frac{1}{2}; - \frac{\omega^2}{4\varepsilon} \right) =
e^{- \frac{\omega^2}{4\varepsilon}} M\left(- \varepsilon  \nu_l^2; \frac{1}{2};  \frac{\omega^2}{4\varepsilon} \right).
\end{equation*}
When $\varepsilon  \nu_l^2 \in \mathbb{N}$, the Kummer function becomes a Laguerre polynomial
\begin{equation*}
L_{n}^{(\alpha)}(x)=\left(\begin{array}{c}
n+\alpha \\
n
\end{array}
\right) M(-n, \alpha+1, x)=\frac{(\alpha+1)_{n}}{n !}{ }_{1} F_{1}(-n, \alpha+1, x)
\end{equation*}
\noindent where $(a)_{n}$ is the Pochhammer symbol, \citep{as64}. The value computed via this formula produces the correct result in Mathematica, which, e.g., shows that $M\left(\frac{1}{2}  + \varepsilon  \nu_l^2; \frac{1}{2}; - \frac{\omega^2}{4\varepsilon} \right)$ doesn't have a singularity when $\varepsilon  \nu_l^2 \in \mathbb{N}$. This can be compared with the direct computation of $M\left(\frac{1}{2}  + \varepsilon  \nu_l^2; \frac{1}{2}; - \frac{\omega^2}{4\varepsilon} \right)$ by replacing it with $M\left(\frac{1}{2} + \epsilon  + \varepsilon  \nu_l^2; \frac{1}{2}; - \frac{\omega^2}{4\varepsilon} \right), \ 0 < \epsilon \ll 1$, because in this case Mathematica provides the correct result.
}.

\end{proof}

On the other hand, the minimum of $F(\omega)$ in variable $\varepsilon$ is reached at $\varepsilon \to 0$. This can be checked by differentiating $F(\omega)$ with respect to $\varepsilon$, setting $\varepsilon$ = 0, and taking into account that
\[   \lim_{\varepsilon \to 0} M\left(\frac{3}{2}, \frac{3}{2}, -\frac{x^2}{4 \varepsilon} \right) = 0.  \]
\begin{figure}[t]
\begin{center}
\hspace*{-0.3in}
\subfloat[]{\includegraphics[width=0.45\textwidth]{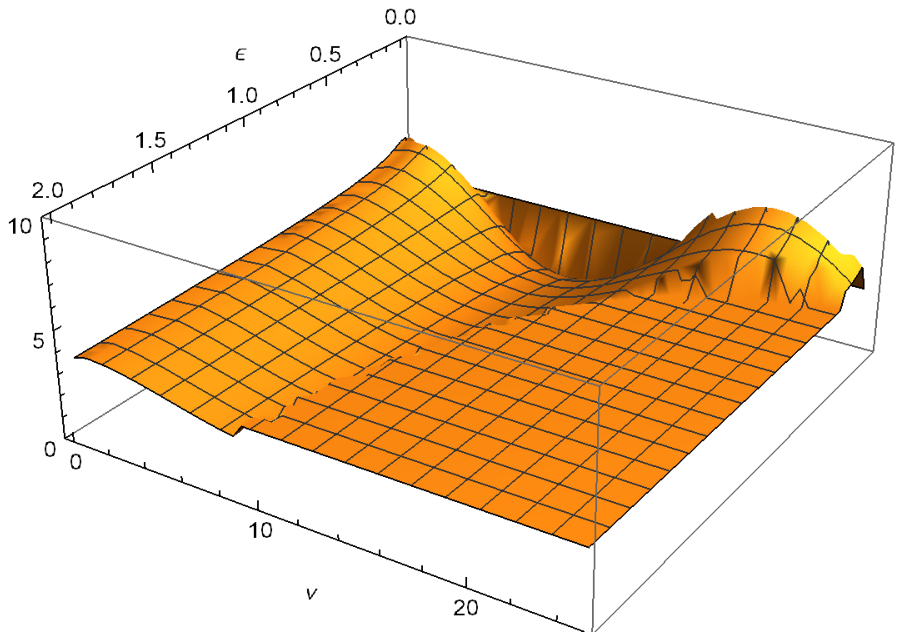}}
\subfloat[]{\includegraphics[width=0.45\textwidth]{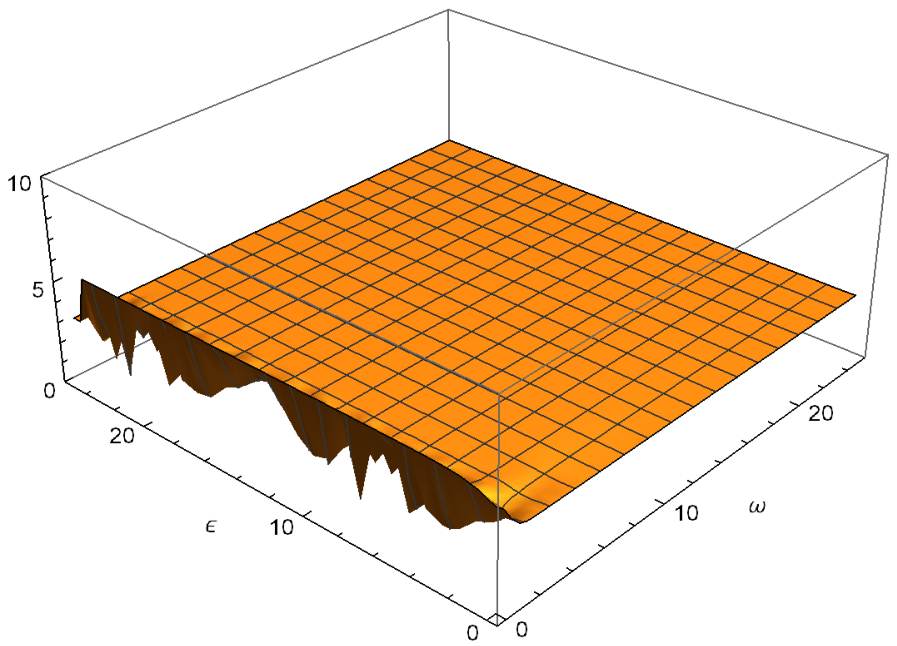}}
\end{center}
\caption{The behavior of the function $F(\omega)$;  (a) - in coordinates $(\varepsilon, \nu_l)$, (b) - in coordinates $(\omega, \nu_l)$ at $w = 8$.}
\label{Fplot}
\end{figure}

Also, it can be checked that other roots of $F_\varepsilon(\omega)$ lie close to zero. Then a typical behavior of $F(\omega)$ is presented in Fig.~\ref{Fplot} in coordinates $(\omega, \nu_l)$ and $(\varepsilon, \nu_l)$ with $w = 8$. One can see that $F(\omega)$ is positive everywhere and tends to zero when its parameters take extreme values. Therefore, $F(\omega)$ is positive. Accordingly, $\bar{\Theta}_{kl}(t,\nu)$ is positive definite.

It worth mentioning that this proof is not 100\% rigorous in its last part, and is relying more on intuitive and practical arguments. Nevertheless, the value of $w$ can always be chosen in a way that makes $F(\omega)$ positive. The exact value of $w$ doesn't influence our final result in \eqref{volt} because the integral on $\nu$ in \eqref{volt} of the Dirac Delta function vanishes at any finite value of the multiplier $w$.

\end{document}